\documentclass[11pt]{article}

\usepackage{subfigure}
\usepackage{natbib}
\usepackage{amsmath}
\usepackage{graphicx}
\usepackage{fullpage}

\usepackage{latexsym}
\usepackage{amsthm}
\usepackage{amsmath}
\usepackage{amstext}
\usepackage{amsfonts}
\usepackage{amsopn}

\usepackage{ifthen}

\newtheorem{theorem}{Theorem}[section]
\newtheorem{lemma}[theorem]{Lemma}
\newtheorem{corollary}[theorem]{Corollary}
\newtheorem{proposition}[theorem]{Proposition}
\newtheorem{definition}{Definition}[section]

\newcommand{\prob}[2][]{\text{\bf Pr}\ifthenelse{\not\equal{}{#1}}{_{#1}}{}\!\left[#2\right]}
\newcommand{\expect}[2][]{\text{\bf E}\ifthenelse{\not\equal{}{#1}}{_{#1}}{}\!\left[#2\right]}
\newcommand{\given}{\ \mid\ }

\newcommand{\setsize}[1]{\left| #1 \right|}

\newcommand{\suchthat}{\ :\ }

\newcommand{\super}[1]{^{(#1)}}

\newcommand{\dist}{F}

\DeclareMathOperator{\RSEM}{RSEM}

\newcommand{\val}{v}
\newcommand{\vals}{{\mathbf \val}}

\newcommand{\wal}{w}
\newcommand{\wals}{{\mathbf \wal}}
\newcommand{\wali}[1][i]{{\wal_{#1}}}

\newcommand{\vali}[1][i]{{\val_{#1}}}
\newcommand{\valith}[1][i]{{\val_{#1}}}

\newcommand{\util}{u}

\newcommand{\utili}[1][i]{{\util_{#1}}}

\newcommand{\VV}[1][]{\Phi\ifthenelse{\not\equal{}{#1}}{^{#1}}{}}
\newcommand{\ivv}[1][]{\bar{\VV}\ifthenelse{\not\equal{}{#1}}{^{#1}}{}}
\newcommand{\eivv}[1][]{\widetilde{\VV}\ifthenelse{\not\equal{}{#1}}{^{#1}}{}}

\newcommand{\price}{p}
\newcommand{\prices}{{\mathbf \price}}
\newcommand{\pricei}[1][i]{{\price_{#1}}}

\newcommand{\alloc}{x}
\newcommand{\allocs}{{\mathbf \alloc}}

\newcommand{\alloci}[1][i]{\alloc_{#1}}

\newcommand{\drop}[2]{{{#2},\vals_{-#1}}}

\newcommand{\feasibles}{{\cal X}}

\DeclareMathOperator{\ICoperator}{IC}
\newcommand{\IC}[1]{\ICoperator\ifthenelse{\not\equal{}{#1}}{^{#1}}{}}

\DeclareMathOperator{\EFoperator}{EF}
\newcommand{\EF}[1]{\EFoperator\ifthenelse{\not\equal{}{#1}}{^{#1}}{}}

\DeclareMathOperator{\ICOoperator}{ICO}
\newcommand{\ICO}[1]{\ICOoperator\ifthenelse{\not\equal{}{#1}}{^{#1}}{}}

\DeclareMathOperator{\EFO}{EFO}

\DeclareMathOperator{\PFOoperator}{PFO}
\newcommand{\PFO}[1][]{\PFOoperator\ifthenelse{\not\equal{}{#1}}{^{#1}}{}}

\DeclareMathOperator{\Roperator}{R}
\newcommand{\R}[1][]{\Roperator\ifthenelse{\not\equal{}{#1}}{^{#1}}{}}

\newcommand{\IR}[1][]{\bar{\Roperator}\ifthenelse{\not\equal{}{#1}}{^{#1}}{}}
\newcommand{\ER}[1][]{\widetilde{\Roperator}\ifthenelse{\not\equal{}{#1}}{^{#1}}{}}
\newcommand{\Rhat}[1][]{\widehat{\Roperator}\ifthenelse{\not\equal{}{#1}}{^{#1}}{}}
\newcommand{\IRs}{\IR[S]}

\newcommand{\vhat}{\hat{v}}
\newcommand{\vhati}[1][i]{\vhat_{#1}}
\newcommand{\vhats}{{\mathbf \vhat}}

\newcommand{\f}{\mathcal{F}}
\newcommand{\ftwo}{\f^{(2)}}

\newcommand{\runder}{\underline{R}^S}
\newcommand{\amscr}{4.68}

\title{Envy Freedom and Prior-free Mechanism Design}

\author{Nikhil R. Devanur 
\and Jason D. Hartline
\and Qiqi Yan}

\newcommand{\multiunitratio}{6.24}
\newcommand{\digitalgoodratio}{3.12}
\newcommand{\RSEMmultiunitratio}{12.5}
\newcommand{\RSEMdownwardclosedratio}{189}
\DeclareMathOperator{\VCGr}{VCG-with-reserve}
\newcommand{\resi}{r}
\newcommand{\imbal}{\lambda}
\newcommand{\highp}{p}
\newcommand{\lowp}{q}
\newcommand{\highi}{i}
\newcommand{\lowi}{j}

\begin{document}

\maketitle

\begin{abstract}
We consider the provision of an abstract service to single-dimensional
agents.  Our model includes position auctions, single-minded
combinatorial auctions, and constrained matching markets.  When the
agents' values are drawn from a distribution, the Bayesian optimal
mechanism is given by \citet{mye-81} as a virtual-surplus optimizer.
We develop a framework for prior-free mechanism design and analysis.
A good mechanism in our framework approximates the optimal mechanism
for the distribution if there is a distribution; moreover, when there
is no distribution this mechanism still performs well.

We define and characterize optimal envy-free outcomes in symmetric
single-dimensional environments.  Our characterization mirrors
Myerson's theory.  Furthermore, unlike in mechanism design where there
is no point-wise optimal mechanism, there is always a point-wise
optimal envy-free outcome. 

Envy-free outcomes and incentive-compatible mechanisms are similar in
structure and performance.  We therefore use the optimal envy-free
revenue as a benchmark for measuring the performance of a prior-free
mechanism.  A good mechanism is one that approximates the envy-free
benchmark on any profile of agent values.  We show that good mechanisms
exist, and in particular, a natural generalization of the random
sampling auction of \citet{GHW-01} is a constant approximation.
\end{abstract}

\section{Introduction}

%
%
The theories of optimal mechanism design for revenue maximization and
social surplus maximization are fundamentally different.  The
revenue-optimal mechanism \citep{mye-81} depends on the prior
distribution from which the values of the agents are drawn whereas the
surplus-optimal mechanism is prior-free \citep{vic-61,cla-71,gro-73}.
In fact, the latter result is singular in this regard.  Incentive
constraints bind across possible misreports of each agent; therefore,
an optimal mechanism must generally trade off performance on one input
(i.e., profile of agent valuations) for another.

%
%
The goal of prior-free mechanism design and this work therein is
to sacrifice optimality to obtain prior freedom.  Of course, the hope
is that not too much in the way of performance is lost.  One way to
make such a hope precise is to require that the prior-free mechanism
perform within a constant factor of the Bayesian optimal mechanism
when there is a distribution.  

%
%
With this goal in mind, \citet{HR-08} give prior-free single-item
auctions to maximize the agents' utility (surplus less payments).  For
any profile of valuations their mechanism approximates the performance
of the best, in hindsight, (i.i.d.) Bayesian optimal mechanism.
Approximation of this prior-free benchmark implies simultaneous
approximation of the Bayesian optimal auction when the values of the
agents are drawn from an i.i.d.~distribution.  Our main goal is to
extend this framework and approach to richer environments that include
position auction \citep[e.g.,][]{var-06,EOS-07}, single-minded combinatorial
auctions \citep[e.g.,][]{LOS-02}, and constrained matching markets
\citep[e.g.,][]{BdVSV-08}.

%
%
A first challenge with our goal is that the benchmark of \citet{HR-08}
is analytically complex in general environments.  Therefore, we
consider an alternative approach motivated by the similarity between
incentive and no-envy constraints.  An outcome is envy free if no agent
prefers the treatment of another to her own.  I.e., no agent wishes to
unilaterally swap outcomes with another.  Incentive compatibility, on
the other hand, requires that no agent wants to misreport her value.
Importantly envy freedom, as a constraint, binds point wise on a
valuation profile whereas incentive compatibility binds across
different valuation profiles (for the possible misreports of the
agents).  While incentive-compatible mechanisms do not generally have
point-wise optima, envy-free outcomes do.  

%
%
A centerpiece of our work is a characterization of envy-free outcomes
and their optima.  This characterization is structurally equivalent to
Myerson's \citeyearpar{mye-81} characterization of Bayesian optimal
mechanisms applied to the empirical distribution given by the actual
profile of agent values.  The envy-free optimal outcome is a
virtual-surplus maximizer.  This connection between envy freedom and
incentive compatibility is implicit in \citet{JK-07} where it is shown
that they are equivalent in the limit, in contrast we show that they
are structurally equivalent generally.

%
%
As a benchmark for prior-free mechanism design, the optimal envy-free
performance is appealing as it inherits many nice properties of
Bayesian optimal mechanisms.  It is the maximum of a linear objective
subject to feasibility constraints.  Furthermore, in environments such
as constrained matching markets, the envy-free optimal revenue
point-wise dominates the revenue of the Bayesian-optimal
incentive-compatible mechanism for any i.i.d.~distribution.  We view
the envy-free optimal revenue as an unattainable prior-free
``first-best'' solution and, because of constraints of incentive
compatibility, attempt to show that ``second-best'' approximation
mechanisms are not to far away from it.  While the \citet{JK-07}
result implies such a result holds in the limit, we show that it
continues to hold in general.

%
%
The incentive-compatible mechanisms that we design will satisfy the
following prior-free performance guarantee: Their revenue on any
profile of valuations will be within a constant multiplicative factor
of the optimal envy-free revenue for the same valuation profile.  The
approximation factors we obtain range from $\multiunitratio$ to
roughly 2500.  While these bounds are outside the realm of being
practically significant in themselves, they do have important
theoretical and practical implications.  A mechanism that is a
constant approximation in theory must exhibit at least some of the
necessary characteristics needed for good performance in practice.  As
an example of the contrapositive, we show that mechanisms based on
reserve prices do not give constant approximations in the environments
that we consider.  Furthermore, since our goal is theoretical
understanding of good properties of mechanisms in environments where
the optimal mechanism is complex, properties and analysis of constant
approximations are provide more abstract understanding.  Finally, the
closer a mechanism is to optimal the more potential there is for the
form of the mechanism to be overly dependent on modeling details that may
not be accurate of practice.  Hence, simple mechanisms that give
constant approximations are often robust and detail free in the
sense of \citet{wil-87}.  Finally, in practice these mechanisms
perform much better than their theoretical guarantees.  For deeper
motivation of approximation in mechanism design see the survey of
\citet{har-11}.

%
%
Our first approach for the design of prior-free mechanisms for general
environments is via reduction.  A digital-good environment (a.k.a.,
unlimited supply) is one where the mechanism has no inter-agent
feasibility constraint.  A multi-unit environment (a.k.a., limited
supply) is one where the mechanism has a constraint on the number of
agents that can be simultaneously served, e.g., multiple units of a
single item.  We given a reduction from multi-unit environments to
digital-good environments that loses at most a factor of two in
approximation factor.  We then give a lossless reduction from position
auction environments and constrained matching environments to
multi-unit environments.  To obtain these reductions we give a
structural characterization of these environments that shows that they
are equivalent, even with respect to approximation.  Given the
$\digitalgoodratio$-approximation for digital goods given by
\citet{II-10}, our reduction implies $\multiunitratio$-approximations
for multi-unit, constrained-matching, and position environments.

%
%
Our second approach is via a generalization of the random sampling
auction of \citet{GHW-01} and \citet{BV-03}.  The auction takes the
following form: The agents are randomly partitioned into a market and
a sample.  The sample is then used for market analysis and its
empirical distribution is calculated.  The optimal auction for the
empirical distribution of the sample is then run on the market.  This
prior-free mechanism is one of the most fundamental, and we extend the
analysis techniques derived for it in digital-good environments to
multi-unit and (more generally) downward-closed environments.  A
downward-closed environment is given by a set system that specifies
which agents can be simultaneously served.  The only constraint placed
on this set systems is that subsets of feasible sets are themselves
feasible.  The approximation factors we obtain are
$\RSEMmultiunitratio$ and $\RSEMdownwardclosedratio$ for multi-unit
(and by the reduction above, constrained-matching and position auction
environments) and downward-closed environments respectively.

\paragraph{Related Work.}

Our work fundamentally relies on the theory of optimal auctions as
defined by \citet{mye-81} and refined by \citet{BR-89}.  In
particular, Myerson shows that Bayesian optimal mechanisms are virtual
surplus optimizers and Bulow and Roberts show that the virtual value
of an agent in this virtual surplus maximization can be viewed as the
marginal revenue as given by an agent's value distribution.

The random sampling auction for digital goods was first studied by
\citet{GHW-01}.  The asymptotic performance of the mechanism was given
by \citet{seg-03} and \citet{BV-03} and the convergence rate was
studied by \citet{BBHM-08}.  In contrast, \citet{GHKSW-06} consider
the non-asymptotic behavior of the random sampling auction and show
that its performance is a (large) constant factor from a prior-free
benchmark that in retrospect coincides with ours.  \citet{AMS-09} give
a nearly tight analysis that shows that the random sampling auction is
a 4.68-approximation (the lower-bound is four).

There have been numerous attempts to design good prior-free mechanisms
for digital goods outside the random sampling paradigm.  Two notable
approaches include an approximate reduction to the ``decision
problem''\footnote{Given a target profit, the decision problem is to construct
a mechanism that obtains that target profit when it is attainable.} by
\citet{GHKSW-06} and an approach based on statistical estimates that
are non-manipulable with high probability by \citet{GH-03a}.
\citet{HM-05} extend the former approach and obtain an approximation
factor of 3.25.  Finally, \citet{II-10} show that a convex combination
of these approaches gives an approximation factor of 3.12.  The
6.24-approximation we obtain is the instantiation of our reduction
with the 3.12 approximation of \citet{II-10}.

The digital good auctions described above were analyzed in comparison
to a natural single-priced benchmark.  \citet{HR-08} suggest that
approximation of a prior-free benchmark should imply approximation of
the Bayesian optimal mechanism for any i.i.d.~distribution.
Benchmarks for which such an implication holds are well grounded in
the classical theory of Bayesian optimal auctions.  They propose the
performance of the best, in hindsight, Bayesian optimal mechanism as a
benchmark. For multi-unit environments, they characterize this
benchmark as two-priced lotteries.  In contrast, our benchmark, the
envy-free optimal revenue, can be viewed as a relaxation of the
Hartline-Roughgarden benchmark that is structurally simpler and
analytically tractable in general downward-closed environments.

Subsequent to our study, \citet{HH-12} generalized the
statistical-estimation-based approach of \citet{GH-03a} to design a
32-approximation of the envy-free benchmark in downward-closed
environments.  A generalization of the
reduction-to-the-decision-problem approach of \citet{GHKSW-06} yields
a 19-approximation \citep{HH-12b}.

\paragraph{Overview.}

In Section~\ref{s:prelim} we describe in detail the our abstract
environment for mechanism design which includes multi-unit, position,
and combintorial auctions.  In Section~\ref{s:ef} we characterize
envy-free outcomes and their optima.  In Section~\ref{s:ic_vs_ef} we
compare incentive-compatible and envy-free revenues.  In
Section~\ref{s:benchmarks} we describe our prior-free design and
analysis framework and give Bayesian justification for using the
envy-free optimal revenue as a prior-free performance benchmark.  In
Section~\ref{s:reduction} we describe a reduction-based approach
wherein we show that, e.g., constrained matching mechanisms reduce to
position auctions which reduce to multi-unit auctions which reduce
(with a loss of a factor of two) to digital-good auctions.  In
Section~\ref{s:random_sampling} we describe a random sampling auction
and analyze this auction in downward-closed environments.

\section{Single-dimensional Environments}

\label{s:prelim}

There are $n\geq 2$ agents. Each agent $i$ has a valuation $\vali$ for
receiving an abstract service.  The {\em valuation profile} is $\vals
= (\vali[1],\ldots,\vali[n])$.  We assume that the agents are indexed
in order of decreasing values, i.e., $\vali \geq \vali[i+1]$.  An
agent $i$ who is served with probability $\alloci$ and charged price
$\pricei$ obtains utility $\utili = \vali \alloci - \pricei$.  {\em
  Individual rationality} requires that $\utili$ be non-negative.

We are allowed to serve certain feasible sets of agents as given by a
set system.  The set system is \emph{downward-closed} in the sense if
a set of agents is feasible, so is any of its subsets.  The empty set
is always feasible.  We allow randomization in two senses (1) the set
system constraint may be randomized (i.e., by convex combination over
set systems) and (2) the set of agents served may be random (by
convex combination over feasible sets).  Notably, randomization in (1)
is given by the environment and randomization in (2) is by our choice
of outcome.  We define an {\em allocation} as a vector $\allocs
=(\alloci[1],\ldots,\alloci[n]) \in [0,1]^n$ where $\alloci$ is the
probability that agent $i$ is served.  An allocation is feasible if it is
the characteristic vector induced by the process above.  The environments
permitted include digital good auctions, multi-unit auctions, position
auction environments, matroid environments, and single-minded
combinatorial auctions.

%
%
We further assume that the feasibility constraint imposed by the
environment is symmetric, i.e., the set of feasible allocations is
closed under permutation.  Digital good, multi-unit auction, and
position auction environments are all symmetric.  Given any
asymmetric environment its corresponding {\em permutation environment}
is obtained by randomly permuting the agents with respect to
the feasibility constraint.  Of special interest for us will be
downward-closed permutation environments and matroid permutation
environments.  By definition these environments are symmetric.

To make the model described above precise, we give the following
formal definitions.

\begin{definition}[single-dimensional environment]
There are $n$ agents denoted $N = \{1,\ldots,n\}$.  The sets of agents
that can be simultaneously served are denoted by $\feasibles \subset
2^N$.  The mechanism may serve $S \subset N$ if $S \in \feasibles$.
\end{definition}

\begin{definition}[downward-closed environment]
A {\em downward-closed environment} is a single-dimensional
environment were all subsets of feasible sets are feasible. I.e., $S
\in \feasibles$ and $T \subset S$ implies that $T \in \feasibles$.
\end{definition}

{\em Single-minded combinatorial} auctions are an example of a
downward-closed environment.  In a single-minded combinatorial auction
each agent $i$ desires a specific bundle (i.e., subset) of a set of
$k$ available items.  Agent $i$ has value $\vali$ for receiving her
entire desired bundle (or a superset of it) and value zero otherwise.
We say an agent is ``served'' if she receives her desired bundle.
Notice that a set of agents $S$ can be simultaneously served (i.e., $S
\in \feasibles$) if their desired bundles are disjoint.  Of course, if
$S$ has disjoint desired bundles and $T \subset S$, then $T$ has
disjoint desired bundles, hence $\feasibles$ is downward closed.

\begin{definition}[matroid environment]
A {\em matroid environment} is a downward-closed environment that
satisfies an additional {\em augmentation property}: there is always
an element in a larger cardinality feasible set that can by feasibly
added to a smaller feasible set.  I.e., given two feasible sets $S,T
\subset \feasibles$ with $|S| < |T|$, there exists an $i \in T
\setminus S$ such that $S \cup \{i\} \in \feasibles$.
\end{definition}

Two key consequences of the augmentation property are that (1) all
maximal feasible sets have the same cardinality and that (2) the
greedy algorithm optimizes surplus, i.e., the sum of the values of
the agents served \citep[see e.g.,][]{oxl-92}.  The greedy algorithm
sorts the agents by their values and then, in this order, greedily
serves each agent if it is feasible to do so given the subset of
agents previously served.  Importantly the greedy algorithm is {\em
  ordinal} in that only the order of values matters in determining the
surplus maximizing set and not the magnitude of the values.  Matroid
environments have a rich history in mechanism design, see e.g.,
\citet{tal-03}, \citet{BdVSV-08}, and \citet{HR-09}.

The example of a matroid environment that is most relevant to auction
theory is that of {\em constrained matching markets}.  In a
constrained matching market each agent $i$ desires one of a subset of
$k$ available items and her value for any item in this subset is
$\vali$ (her value for any other item is zero).  The agent demand sets
induce a bipartite graph between agents and items where an edge is
between agent $i$ and item $j$ if $j$ is one of $i$'s desired items.
A subset $S$ of agents can be simultaneously served if there is a
matching in the bipartite graph were all of $S$ is matched.  This
matroid is known as the {\em transversal matroid}.

A special case of the transversal matroid is when the market is
essentially partitioned into $\ell$ parts and within each part every
item is acceptable to every agent (of course, there may be fewer items
than agents in the part).  This is known as the {\em partition
  matroid}.  A special case if the partition matroid is the one where
there is only one part, i.e., there are $k$ identical items, and $n$
agents who each desire one of these items.  This matroid is known as
the {\em $k$-uniform matroid} and it corresponds precisely to the
standard $k$-unit auction environment.

\begin{definition}[multi-unit environment]
There are $n$ unit-demand agents and $k$ units of an item.  I.e.,
$\feasibles = \{ S \suchthat |S| \leq k \}$.
\end{definition}

An important special case of the multi-unit environment is one where
the supply constraint $k$ is never binding, i.e., when $k = n$.  This
special case is known as the digital-good environment.

\begin{definition}[digital-good environment]
There are $n$ unit-demand agents denoted $N = \{1,\ldots,n\}$ and any
subset of them can be served.  I.e., $\feasibles = 2^N$.
\end{definition}

A generalization of multi-unit environments that has recently been
under intense scrutiny due to its application to auctions for selling
advertisements on Internet search auction is the {\em position
  environment} \citep{var-06,EOS-07}.

\begin{definition}[position environment]
There are $n$ agents and $n$ positions.  The positions have
non-increasing weights $\wals = (\wali[1],\ldots,\wali[n])$.  If an
agent $i$ is assigned position $j$ she served with probability
$\wali[j]$ and her value for this assignment is $\vali\wali[j]$.
\end{definition}

It has been observed, e.g., by \citet{Dughmi09}, that there is a close
connection between position environments and convex combinations of
multi-unit environments.  E.g., if we choose the supply $k$ randomly
so that $k$-units are available with probability $\wali[k] -
\wali[k+1]$ then this distribution over $k$-unit environments gives
the same service probabilities as the position auction.  Of course, a
$k$-unit auction is a position auction wherein the $k$ highest weights
are one, and the remaining weights are zero.

Digital-good, multi-unit, and position environments are
agent-symmetric whereas matroid and downward-closed environments may
not be.  We can compose any environment with a permutation that
renames the agents with respect to the set system.  The convex
combination of these environments with the permutation taken uniformly
at random from the set of $n$-element permutations is referred to as a
{\em permutation environment} and it is symmetric.  

\begin{definition}
Given any $n$ agent single-dimensional environment specified by
$\feasibles$, the permutation environment for $\feasibles$ draws
permutation $\pi$ (mapping $i$ to $j$ by $\pi(i) = j$) uniformly at
random from the set of all $n$-element permutations, and the feasible
sets are given by $\feasibles_\pi = \{ \{\pi(i) \suchthat i
\in S\} \suchthat S \in \feasibles\}$.
\end{definition}

\section{Optimal Envy-free Pricing}

\label{s:ef}

In this section we derive a theory of optimal envy-free outcomes in
single-dimensional environments that mirrors that of Bayesian optimal
(incentive compatible) mechanisms for i.i.d.~prior
distributions~\citep[cf.,][]{mye-81,BR-89}.

%
%
\begin{definition}[Envy freedom]
An allocation $\allocs$ with payments $\prices$ is {\em envy free} for
valuation profile $\vals$ if no agent prefers the outcome of another
agent to her own.  Formally,
\begin{align*}
\forall i, j,\ \vali \alloci - \pricei
     &\geq \vali \alloci[j] - \pricei[j].
\end{align*}
\end{definition}

We first characterize envy-free outcomes in terms of the allocation.
For a given allocation $\allocs$ there may be several pricings
$\prices$ for which the allocation is envy free.  Since our objective
is profit maximization we will characterize the $\prices$
corresponding to $\allocs$ that gives the highest total revenue.  The
proof of this characterization as it is nearly identical to that of
the analogous (and standard) characterization of incentive compatible
mechanisms; we include it for completeness.

\begin{definition}
An allocation is {\em swap monotone} if the allocation probabilities
have the same order as the valuations of the agents, i.e., $\alloci
\geq \alloci[i+1]$ for all $i$. (Recall agents are ordered with $\vali \geq \vali[i+1]$.)
\end{definition}

\begin{lemma}\label{l:ef_char} In symmetric environments, 
an allocation $\allocs$ admits a non-negative and individual rational
payment $\prices$ such that $(\allocs,\prices)$ is envy free if
and only if $\allocs$ is swap monotone.  If $\allocs$ is swap
monotone, then the maximum payments for which $(\allocs,\prices)$ is
envy free satisfy, for all $i$, (Figure~\ref{fig:ef_payment})
\begin{align*}
\pricei 
&= \sum\nolimits_{j=i}^n (\vali[j]-\vali[j+1])\cdot (\alloci - \alloci[j+1])\\
&= \sum\nolimits_{j=i}^n \vali[j] \cdot (\alloci[j] - \alloci[j+1]).
\end{align*}
\end{lemma}

\begin{figure*}[tbp]
\centering
  \includegraphics[scale=0.8]{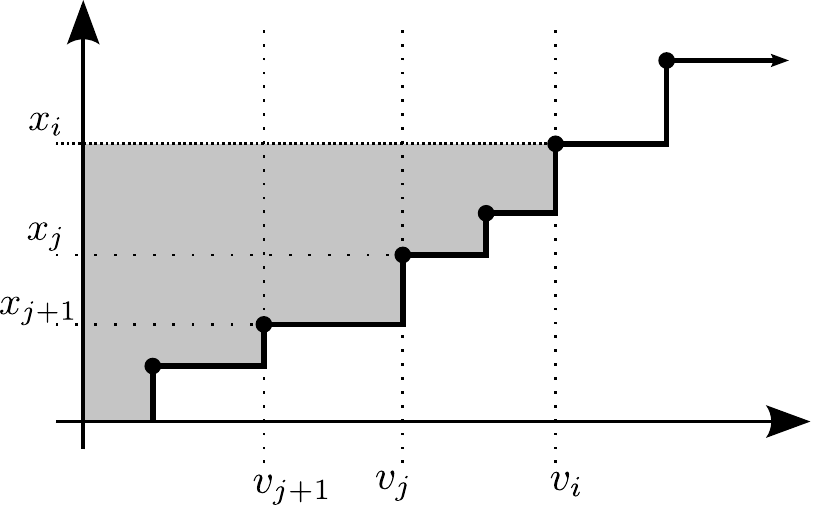}
\caption{The solid curve depicts a swap monotone allocation as a function of
  the values (points). The shaded area corresponds to agent $i$'s payment $\pricei$ from Lemma~\ref{l:ef_char}.}
\label{fig:ef_payment}
\end{figure*}

\begin{proof}
Monotonicity and the payment identity imply envy-freedom:  Suppose
$\allocs$ is swap monotone. Let $\prices$ be given as in the lemma. We
verify that $(\allocs,\prices)$ is envy-free. There are two cases: if
$i\leq j$, we have:
\begin{align*}
\pricei-\pricei[j] &= \sum_{k=i}^{j-1}v_{k}\cdot(\alloci[k]-\alloci[k+1])
 \leq \vali\cdot\sum_{k=i}^{j-1}(\alloci[k]-\alloci[k+1])
 = \vali\cdot(\alloci[i]-\alloci[j]),\\
\intertext{and if $i\geq j$, we have:}
\pricei-\pricei[j] &=-\sum_{k=j}^{i-1}v_{k}\cdot(\alloci[k]-\alloci[k+1])
 \leq -\vali\cdot\sum_{k=j}^{i-1}(\alloci[k]-\alloci[k+1])
 = \vali\cdot(\alloci[i]-\alloci[j]).
 \end{align*}
Rearranging the results of these calculations we have the definition
of envy freedom.

Envy-freedom implies monotonicity:  Suppose $\allocs$ admits $\prices$
such that $(\allocs,\prices)$ is envy-free. By definition, $\vali
\alloci - \pricei \geq \vali \alloci[j] - \pricei[j]$ and $\vali[j]
\alloci[j] - \pricei[j] \geq \vali[j] \alloci - \pricei$.  By summing
these two inequalities and rearranging,
$(\alloci-\alloci[j])\cdot(\vali-\vali[j])\geq0$, and hence $\allocs$
is swap monotone.

The maximum envy-free prices satisfy the payment identity: Agent $i$
does not envy $i+1$ so $\vali \alloci - \pricei \geq \vali
\alloci[i+1] - \pricei[i+1]$, or rearranging: $\pricei \leq
\vali(\alloci - \alloci[i+1]) + \pricei[i+1]$.  Given $\pricei[i+1]$
the maximum $\pricei$ satisfies this inequality with equality.
Letting $\pricei[n] = \vali[n]\alloci[n]$ (the maximum individually
rational payment) and induction gives the payment identity: $\pricei =
\sum^n_{j=i}\vali[j]\cdot(\alloci[j]- \alloci[j+1])$.
\end{proof}

Importantly, the above characterization leaves us free to speak of the
(maximum) envy-free revenue of any swap monotone allocation $\allocs$ on values
$\vals$, which we denote by $\EF{\allocs}(\vals)$.  For any $\vals$
and any symmetric environment we will now solve for the envy-free optimal
revenue, denoted by $\EFO(\vals)$.

\begin{figure*}[htbp]
\centering
	\includegraphics[scale=0.5]{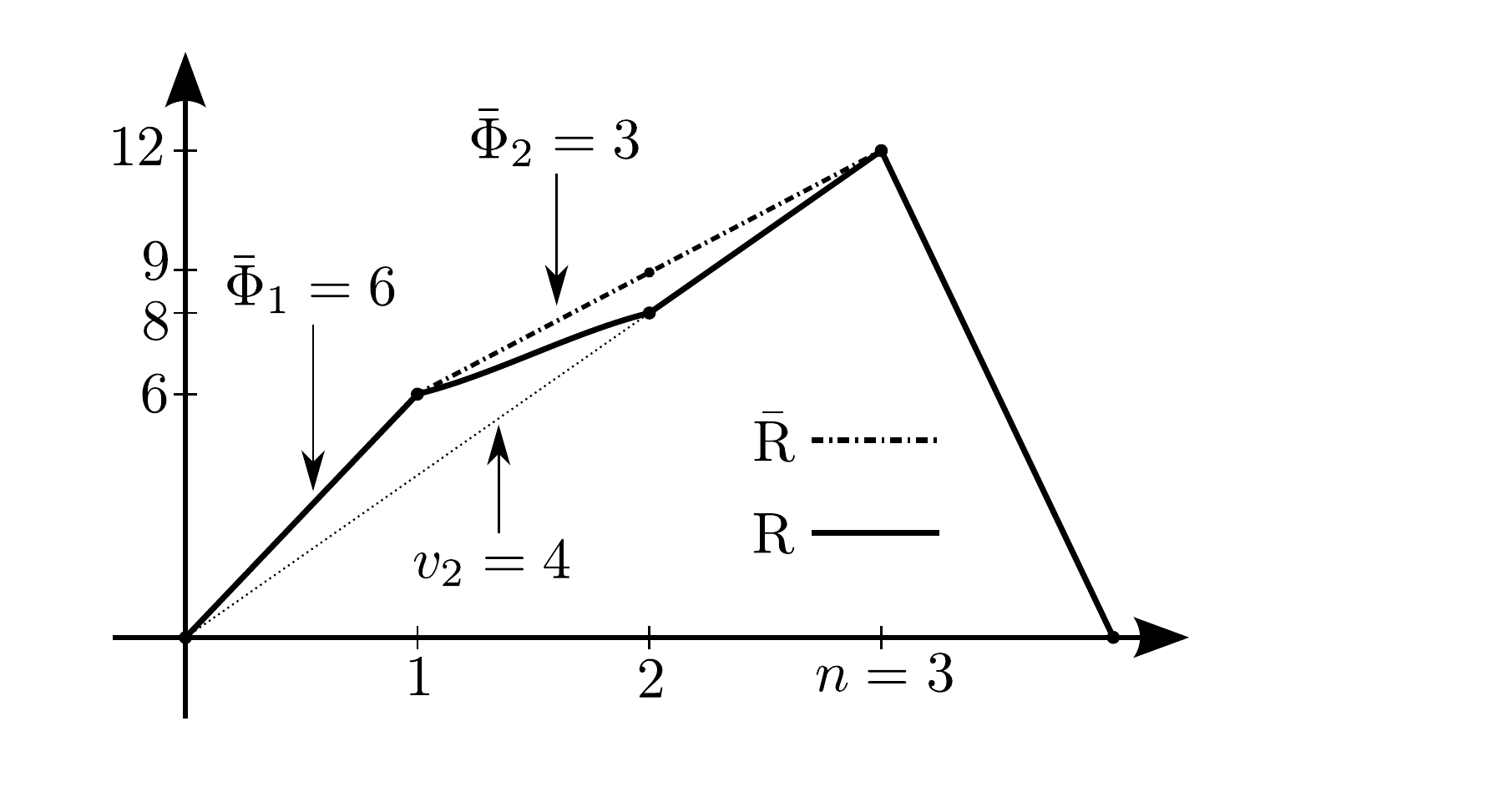}
	\caption{$\R$ and $\IR$ are the revenue curve and ironed
          revenue curve of the valuation profile $(6,4,4)$. The ironed
          virtual value of the high-value agent is $6$, and the ironed
          virtual value of the two low-value agents are both
          $(12-6)/2=3$.  E.g., the optimal EF revenue in the $k=2$ unit environment is
          $\IR(2) = 9$.}
	\label{fig:ironing}
\end{figure*}

We will characterize the envy-free optimal revenue in terms of
properties of the valuation profile $\vals$.  Given a valuation
profile $\vals$ we denote the {\em revenue curve} by $\R[\vals](i) =
i\cdot \vali$ for $i=\{1,\ldots,n\}$ (recall $\vali$'s are indexed in
decreasing order).  For convenience we also let
$\R[\vals](0)=\R[\vals](n+1)=0$.  The {\em ironed revenue curve},
denoted $\IR[\vals](i)$, is the minimum concave function that
upper-bounds $\R[\vals]$.  Likewise, define the {\em virtual valuation
  function} $\VV[\vals](\val) = \R[\vals](i) - \R[\vals](i-1)$ and the
{\em ironed virtual valuation function} $\ivv[\vals](\val) =
\IR[\vals](i) - \IR[\vals](i-1)$, where $i\in\{1,\ldots, n+1\}$ is
such that $\val \in [\vali,\vali[i-1])$.  (We set $\vali[0]=\infty$
  for notational convenience.)  See Figure~\ref{fig:ironing}.

$\R[\vals](i)$ is the best envy-free revenue one can get from serving
  exactly $i$ agents at the same price deterministically.  Consider a
  2-unit auction example with one high-value
  agent with value 6 and two low-value agents with value 4. It is
  envy free to serve one high-value agent and one low-value agent at
  price $4$, achieving revenue $\R(2)=8$. Interestingly, this is not
  optimal.
The following allocation and payments are also envy-free: serve the
high-value agent with probability 1 at price 5, and serve 
a low-value agent chosen at random at price 4.
Both units are always sold and the total revenue is $\IR(2)=9$.  
In what follows we will derive that
this revenue is optimal among all
envy-free outcomes.



\begin{lemma}
\label{l:vv=rev} The (maximum) envy-free revenue of a swap monotone allocation $\allocs$ satisfies:
$$ \EF{\allocs}(\vals) 
= \sum\nolimits_{i=1}^n \R[\vals](i) \cdot (\alloci - \alloci[i+1])
= \sum\nolimits_{i=1}^n \VV[\vals](\vali) \cdot \alloci.
$$
\end{lemma}

\begin{proof}

The proof is by the following equalities:
\begin{align*}
\EF{\allocs}(\vals) & =  \sum\nolimits_{i=1}^n \pricei = \sum\nolimits_{i=1}^n\sum\nolimits_{j=i}^n\vali[j]\cdot(\alloci[j]-\alloci[j+1])\\
 & = \sum\nolimits_{i=1}^ni\vali\cdot(\alloci-\alloci[i+1])
  =  \sum\nolimits_{i=1}^n\R(i)\cdot(\alloci-\alloci[i+1])\\
 & =  \sum\nolimits_{i=1}^n(\R(i)-\R(i-1))\cdot \alloci
  =  \sum\nolimits_{i=1}^n\VV[\vals](\vali)\cdot \alloci. \qedhere
 \end{align*}
\end{proof}

An implication of the characterization of the envy-free revenue of a
pricing as its {\em virtual surplus}, i.e., $\sum_i \VV(\vali)
\alloci$, suggests that to maximize revenue, the allocation should
maximize virtual surplus subject to swap monotonicity (and
feasibility).  In symmetric environments with monotone virtual
valuation functions, the maximization of virtual surplus results in a
swap monotone allocation.  In general symmetric environments, the
allocation that maximizes \emph{ironed} virtual surplus is both swap
monotone and revenue optimal among all swap monotone allocations.

\begin{lemma}\label{l:swap_mon}
In a symmetric environments, the allocation that maximizes ironed
virtual surplus with ties broken randomly is swap monotone.
\end{lemma}

\begin{proof} 
Suppose $\ivv(\vali) > \ivv(\vali[j])$ then $\alloci \geq \alloci[j]$;
otherwise, swapping $\alloci$ for $\alloci[j]$ would have higher
ironed virtual surplus.  Suppose $\ivv(\vali) = \ivv(\vali[j])$, then
$\alloci = \alloci[j]$ because of random tie breaking and the
symmetry of the environment.
\end{proof}

\begin{theorem}
\label{t:efo}
In any symmetric environment with any valuation profile $\vals$, the
allocation $\allocs$ that maximizes ironed virtual surplus
w.r.t.~$\ivv[\vals]$ maximizes envy-free revenue among all
swap-monotone allocations. I.e., $\EFO(\vals) = \EF{\allocs}(\vals)$.
\end{theorem}

This theorem is proved by a useful lemma that relates revenue to
ironed virtual surplus.

\begin{lemma}
\label{l:ironing}
 For any swap-monotone allocation $\allocs$ on valuation profile
 $\vals$,
$$
\EF{\allocs}(\vals) 
\leq \sum\nolimits_{i=1}^n \ivv[\vals](\vali) \cdot \alloci
=    \sum\nolimits_{i=1}^n\IR[\vals](i)\cdot(\alloci[i]-\alloci[i+1]),
$$ 
with equality holding if and only if $\alloci = \alloci[i+1]$
whenever $\IR[\vals](i) > \R[\vals](i)$.
\end{lemma}

\begin{proof}
To show the inequality, we have:
\begin{align*}
\EF{\allocs}(\vals) &= \sum\nolimits_{i=1}^n\R(i)\cdot(\alloci[i]-\alloci[i+1])\\
 &=  \sum\nolimits_{i=1}^n\IR(i)\cdot(\alloci[i]-\alloci[i+1])\\
 &\quad -\sum\nolimits_{i=1}^n(\IR(i)-\R(i))\cdot(\alloci[i]-\alloci[i+1])\\
 &\leq  \sum\nolimits_{i=1}^n\IR(i)\cdot(\alloci[i]-\alloci[i+1]),
 \end{align*}
where we use the fact that $\IR(i)\geq \R(i)$ and $\alloci[i]\geq \alloci[i+1]$.
Clearly the equality holds if and only if  $\alloci[i]=\alloci[i+1]$
whenever $\IR(i)>\R(i)$.
\end{proof}

\begin{proof}[Proof of Theorem~\ref{t:efo}]
Consider $\allocs$ that optimizes ironed virtual surplus with random
tie breaking and also consider any other swap monotone $\allocs'$.
Note that whenever $\IR(i)>\R(i)$, we have
$\ivv[\vals](\vali)=\ivv[\vals](\vali[i+1])$ for which random
tie-breaking implies $\alloci[i]=\alloci[i+1]$.  Therefore $\allocs$
satisfies Lemma~\ref{l:ironing} with equality and it is optimal for
the summation of the equality, whereas $\allocs'$ satisfies it with
inequality and may not be optimal for the summation.  Thus
$\EF{\allocs}(\vals) \geq \EF{\allocs'}(\vals)$ and $\allocs$ is
revenue optimal.
\end{proof}

%
%
As an example of this theorem, consider the position auction
environment with weights $w_1 \geq w_2 \geq \ldots \geq w_n$.  An
ironed virtual surplus maximizer assigns agents with higher ironed
virtual values to slots with larger click probabilities, breaking ties
randomly, ignoring agents with negative ironed virtual values. The
ironed virtual surplus, and thus revenue, is $\sum_{\{i \suchthat
  \ivv(\vali) \geq 0\}} \ivv(\vali) \cdot w_i$, which can be read off
the revenue curve, e.g., Figure~\ref{fig:ironing}.

Importantly, ironed virtual surplus maximization for position auctions
is ordinal, i.e., only the order of the ironed virtual values matters.
The optimal envy-free outcome can then rephrased as follows: First,
tentatively assign the agents to slots in order of their values.
Second, randomly permute the order of each group of agents with
equal ironed virtual surplus.  In section~\ref{s:reduction} we will
discuss consequences for environments for which surplus maximization
is ordinal.

In downward-closed permutation environments, after the set system is
realized, we find the allocation that maximizes ironed virtual
surplus.  With the appropriate payments, this outcome is envy free
given the permutation.  Importantly, this optimization is not ordinal.

\section{Incentive Compatibility versus Envy Freedom}

\label{s:ic_vs_ef}

The major challenge in designing and analyzing incentive compatible
mechanisms is that the incentive constraint binds across all possible
misreports of the agents.  We therefore view a mechanism as an
allocation rule and payment rule pair where $\allocs(\vals)$ and
$\prices(\vals)$ denote the allocation and payments as a function of
the agent values.

%
%
\begin{definition}[Incentive Compatibility]
A mechanism is {\em incentive compatible} if no agent prefers the
outcome when misreporting her value to the outcome when reporting the
truth.  Formally,
\begin{align*}
\forall i, z, \vals, \quad  \vali \alloci(\vals) - \pricei(\vals)
     &\geq \vali \alloci(\drop{i}{z}) - \pricei(\drop{i}{z}),
\end{align*}
where $(\drop{i}{z})$ is obtained from $\vals$ with
$\vali$ replaced by $z$.
\end{definition}

\begin{definition}[Value Monotonicity]
An allocation rule is {\em value monotone} if the probability that an
agent is served is monotone non-decreasing in her value, i.e.,
$\alloci(\drop{i}{z})$ is non-decreasing in $z$ for all agents $i$.
\end{definition}

The following well-known theorem characterizes ex post IC mechanisms.

\begin{theorem} \citep{mye-81}
\label{t:IC_payment}
An allocation rule $\allocs(\cdot)$ admits a non-negative and individually
rational payment rule $\prices(\cdot)$ such that $(\allocs,\prices)$ is
incentive compatible if and only if $\allocs(\cdot)$ is value monotone, and
the uniquely determined payment rule is:
$$\pricei(\vals) = \vali \alloci(\vals) - \int_0^{\vali}
\alloci(\drop{i}{z}) dz.$$
\end{theorem}

Because the payments are uniquely determined by the allocation rule,
for any allocation rule $\allocs(\cdot)$, we let $\IC{\allocs}(\vals)$
denote the IC revenue from running $\allocs(\cdot)$ over $\vals$.

We now compare envy-free revenue to incentive-compatible revenue for
ironed virtual surplus optimizers in permutation environments, i.e.,
where agents are assigned to roles in the set system via a random
permutation.  We show that these quantities are often within a factor
of two of each other.


First we lower bound IC revenue by half of the maximum envy-free
revenue under a technical condition.  In the following we use
$\IC{\ivv}_i(\vals)$ and $\EF{\ivv}_i(\vals)$ to denote the IC and EF
revenue from agent $i$ by applying the ironed virtual surplus
maximizer $\ivv$, respectively.

\begin{lemma}
\label{l:IC>EF/2}
For downward-closed permutation environments, all valuations $\vals$,
and $\ivv$, the ironed virtual valuation function corresponding to some
$\vals'$ obtained from $\vals$ by setting a subset of agents' values
to be 0, we have that $\IC{\ivv}_i(\vals) \geq \tfrac{1}{2}
\EF{\ivv}_i(\vals)$ for all $i$.
\end{lemma}

\begin{proof} 
Let $\allocs(\cdot)$ denote the allocation rule of the ironed virtual
surplus optimizer $\ivv$.  By the assumption of the lemma, for all
$j$, $\ivv(z)$ is constant for all $z\in [\vali[j+1],\vali[j])$, and
  hence the IC allocation rule in fact maps each $z\in
  [\vali[j+1],\vali[j])$ to $\alloci[i](\drop{i}{\vali[j+1]})$.

By Lemma~\ref{t:IC_payment}, $\IC{\ivv}_i(\vals)$ is equal to
$\sum^n_{j=i}(\vali[j]-\vali[j+1])\cdot(\alloci[i](\vals)-\alloci[i](\drop{i}{\vali[j+1]}))$ which, referring to Figure~\ref{fig:ic_vs_ef}, 
equals the area above the IC curve and below the horizontal dotted
line.  On the other hand, $\EF{\ivv}_i(\vals)$ is equal to
$\sum^n_{j=i}(\vali[j]-\vali[j+1])\cdot(\alloci[i](\vals)-\alloci[j+1](\vals))$,
which similarly corresponds to the area above the EF curve and below
the horizontal dotted line.  It suffices to prove that:
$\alloci[i](\vals)-\alloci(\drop{i}{\vali[j+1]}) \geq \frac{1}{2}
\cdot (\alloci[i](\vals)-\alloci[j+1](\vals))$. 
Note that $\alloci(\drop{i}{\vali[j+1]})=\alloci[j+1](\drop{i}{\vali[j+1]})$ as
now agents $i$ and $j+1$ have the same value, this is equivalent to
$\alloci[i](\vals)+\alloci[j+1](\vals)\geq
\alloci(\drop{i}{\vali[j+1]})+\alloci[j+1](\drop{i}{\vali[j+1]})$.

The last inequality says that the total winning probability of agent $i$ and $j+1$
can only decrease if agent $i$ lowers her bid to $\vali[j+1]$.
To prove this, we fix the permutation that maps agents to roles of the set system,
and show that the number of winning 
agents from $i$ and $j+1$ can only be lower after agent $i$ decreases her value.
There are two cases to verify:
(1) Agent $i$ wins after the decrease. 
Then before the decrease, agent $i$ had higher value, 
and the optimal feasible set would be the same.
(2) Agent $j+1$ wins and agent $i$ loses after the decrease.
Then before the decrease, at least one of agents $i$ and $j+1$ would win.
\end{proof}

\begin{figure}[htbp]
	\begin{center}\vspace{-.2in}
		\includegraphics[scale=0.8]{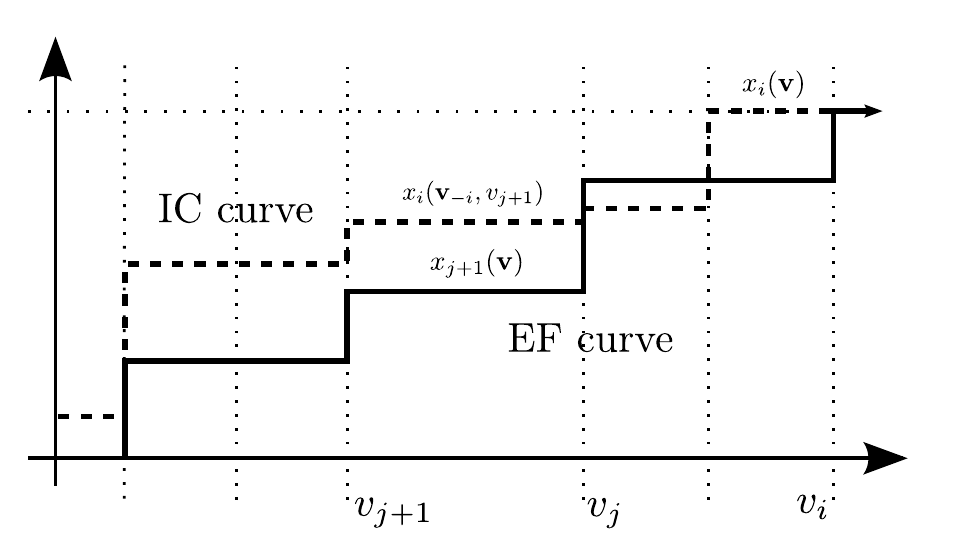}
        \end{center} \vspace{-.2in}
	\caption{Depiction of EF allocation and IC allocation rule
          from which the payments for agent $i$ are computed.  The EF
          allocation curve maps each value in $[\vali[j+1],\vali[j])$
            to $\alloci[j+1](\vals)$, and the IC allocation curve maps
            each $z$ to $\alloci[i](\drop{i}{z})$.}
	\label{fig:ic_vs_ef}
\end{figure}

In matroid permutation environments, envy-free 
revenue upper-bounds incentive-compatible revenue.

\begin{lemma}
\label{l:EF>IC}
For matroid permutation environments, all valuations $\vals$, and all ironed
virtual valuation functions $\ivv$, for all agent $i$, $\EF{\ivv}_i(\vals) \geq
\IC{\ivv}_i(\vals)$.
\end{lemma}

\begin{proof}
Recall that $\EF{\ivv}_i(\vals) = \sum^n_{j=i}(\vali[j]-\vali[j+1])\cdot(\alloci[i](\vals)-\alloci[j+1](\vals))$
and $\IC{\ivv}_i(\vals) = \int_0^{\vali}
(\alloci(\vals)-\alloci(\drop{i}{z})) dz$.  By the monotonicity of
$\alloci(\drop{i}{z})$ in $z$, $\IC{\ivv}_i(\vals)$ is upper-bounded
by
$\sum^n_{j=i}(\vali[j]-\vali[j+1])\cdot(\alloci[i](\vals)-\alloci(\drop{i}{\vali[j+1]}))$.
Recall that
$\alloci(\drop{i}{\vali[j+1]})=\alloci[j+1](\drop{i}{\vali[j+1]})$.
It suffices to prove that $\alloci[j+1](\vals)\leq
\alloci[j+1](\drop{i}{\vali[j+1]})$.  To see this, ironed virtual
surplus maximizers are greedy algorithms in matroid permutation
settings, and if agent $i$ decreases her bid to $\vali[j+1]$, agent
$j+1$ is less likely to be blocked by $i$ who was earlier in the
greedy order, and is hence more likely to be allocated.
\end{proof}

There are downward-closed permutation environments where the envy-free
optimal revenue does not upper-bound the incentive-compatible revenue
of all virtual surplus maximizers.  The proof is by example and can be
found in Appendix~\ref{app:IC>EF}; for this example the amount that the IC
revenue exceeds the EF revenue is a small fraction of the total EF
revenue.

\begin{lemma}\label{l:IC>EF}
There exists a downward-closed permutation environment and valuation
profile $\vals$, such that if $\ivv=\ivv^{\vals}$ is the ironed
virtual valuation function of $\vals$, then $\IC{\ivv}(\vals) >
\EF{\ivv}(\vals)$.
\end{lemma}

\section{Prior-free Mechanism Design and Benchmarks}
\label{s:benchmarks}

%
%
As discussed previously, no incentive-compatible mechanism obtains an
optimal profit point-wise on all possible valuation profiles.
Therefore, to obtain point-wise guarantees, the literature on
prior-free mechanism design looks for the incentive compatible
mechanism that minimizes, over valuation profiles, its worst-case
ratio to a given performance benchmark.  It is important to identify a
good benchmark for such an analysis to be meaningful.

%
%
If the designer had a prior distribution over the agent
valuations then she could design the mechanism that maximizes revenue in
expectation over this distribution.  This is the approach of Bayesian
optimal mechanism design as characterized by \citet{mye-81} and
refined by \citet{BR-89}.  Given a distribution
$\dist$, virtual values and revenue curves can be derived.  The
optimal mechanism is the one that maximizes ironed virtual surplus.

\begin{theorem}
\label{t:mye}
 \citep{mye-81}
When values are i.i.d.~from
distribution $\dist$ the optimal mechanism, $\ICO{\dist}$, is the
ironed virtual surplus optimizer for $\ivv$ corresponding to $\dist$.
\end{theorem}

%
%
If the agent values are indeed drawn from a prior distribution, but
the designer is unaware of the distribution, then a reasonable
objective might be to design a mechanism that is a good approximation
to the optimal mechanism for any unknown distribution.  This {\em
  prior-independent} objective is a relaxation of our
prior-free objective.

%
%
One important criterion for a prior-free benchmark is that its
approximation should imply prior-independent approximation: if a
mechanism is a constant approximation to the benchmark, then for a
relevant class of distributions, it should be a constant approximation
to the Bayesian optimal mechanism under any distribution from the class.

%
%
For matroid permutation environments, Lemma~\ref{l:EF>IC} implies that
for any values $\vals$ the optimal envy-free revenue $\EFO(\vals)$
(which is at least the envy-free revenue of any ironed virtual surplus
optimizer) is at least the incentive compatible revenue of any ironed
virtual surplus optimizer.  By Theorem~\ref{t:mye}, the Bayesian
optimal mechanism is an ironed virtual surplus optimizer so
$\EFO(\vals)$ upper-bounds its revenue.  Consequently, a prior-free
$\beta$-approximation to $\EFO$ is also a prior-independent
$\beta$-approximation for all distributions.

%
%
Unfortunately, even for simple the digital good environment it is not
possible to obtain a prior-free constant approximation to $\EFO$
\citep[see][]{G+06}.  This impossibility arises because it is not possible to
approximate the highest value $\vali[1]$.  For essentially the same
reason, it is not possible to design a prior-independent constant
approximation for all distributions.  We therefore restrict attention
to the large family of distributions with tails that are not too
irregular.  

%
%
\begin{definition}[Tail Regularity] 
A distribution $\dist$ is {\em $n$-tail regular} if in $n$-agent
1-unit environments, the expected revenue of the Vickrey auction is a
2-approximation to that of the Bayesian optimal mechanism.
\end{definition}

The definition of tail regularity is implied by Myerson's regularity
assumption via the main theorem of \citet{BK96}.  The intuition
for the definition is the following. For $n$-agent 1-unit
environments, all the action happens in the tail of the distribution,
i.e, values $\val$ for which $\dist(\val) \approx 1-1/n$; therefore,
irregularity of the rest of the distribution does not have much
consequence on revenue.  Tail regularity, then, restates the
Bulow-Klemperer consequence, as a constraint on the tail of the
distribution and leaves the rest  unconstrained.

We now define the benchmark for prior-free mechanism design.
Approximation of this benchmark guarantees prior-independent
approximation of all $n$-tail-regular distributions.

\begin{definition} \label{d:EFO2}
The {\em envy-free benchmark} is  
$\EFO\super 2(\vals) = \EFO(\vals \super 2)$ where $\vals\super 2 =
  (\vali[2],\vali[2],\vali[3],\ldots,\vali[n])$.
\end{definition}

\begin{theorem}\label{t:tail-regular}
With any $n$-agent matroid permutation environment, any
$n$-tail-regular distribution $\dist$, and any $\beta$-approximation
mechanism to $\EFO \super 2$, the expected revenue of the mechanism
with valuations $\vals$ drawn i.i.d.~from $\dist$ is a
$3\beta$-approximation to the optimal mechanism for $\dist$.
\end{theorem}

\begin{proof}[Proof sketch]
(The full proof in Appendix~\ref{app:tail-regular}.)  We focus on
  showing the result for $k$-unit auctions; the reduction in
  Section~\ref{s:reduction} will enable us to easily generalize this
  to matroid permutation environments.  We use tail regularity to get
  a bound on the payment from the highest agent in terms of the
  Vickrey auction revenue, $\vali[2]$.  Of course, $\EFO \super
  2(\vals)$ is at least $\vali[2]$, and so the payment from the
  highest agent is at most $2 \EFO \super 2(\vals)$ (in expectation
  over i.i.d.~draws of $\vals$ from $\dist$).  The second part of the
  argument involves bounding the total payments of agents
  $\{2,\ldots,n\}$, point-wise from above, by $\EFO \super 2(\vals)$.
  This is possible by Lemma~\ref{l:EF>IC} and detailed analysis of
  $k$-unit auction payments.
\end{proof}

%
%
It is useful to compare the EFO benchmark to ones proposed in the
literature that are based on the Vickrey-Clarke-Groves (VCG) mechanism
with the best (for the particular valuation profile $\vals$) reserve
price \citep[e.g.,][]{GHKSW-06,HR-09}.  The VCG mechanism with a reserve price
first rejects all agents whose values to not meet the reserve, then it
serves the remaining agents to maximize the surplus (sum of values).

The VCG-with-reserve benchmark can be expressed as an ironed virtual
surplus optimizer, and so by Lemma~\ref{l:EF>IC}, in matroid
permutation environments, EFO is no smaller.  For a digital good, EFO
and VCG-with-reserve are identical.  For multi-unit auctions EFO is at
most twice VCG-with-reserve.  For matroid permutation environments EFO
can be (almost) a logarithmic factor larger than VCG-with-reserve.
Therefore, for general environments the EFO-based benchmark results in
stronger approximation guarantees.

\begin{theorem}
\label{thm.benchmarks}
\label{thm:EFO<2F}
For any multi-unit environment and valuation profile $\vals$,
$\EFO(\vals) \leq 2 \VCGr(\vals)$; furthermore,
for $k$ units there exists a valuation profile $\vals$ such that
$\EFO(\vals) \geq (2-\frac{1}{k})\VCGr(\vals)$.
\end{theorem}

\begin{proof}
In the $k$-unit environment, the VCG-with-reserve benchmark on $\vals$
is simply $\max_{i \leq k} \R(i)$.  The envy-free benchmark is
$\EFO(\vals) = \max_{i \leq k} \IR(i)$.  Recall that $\IR$ is concave.
If $\IR$ attains its maximum at $\resi < k$ then $\R(\resi) =
\IR(\resi)$ and both benchmarks obtain the same revenue and the first
part of the lemma holds; so assume that $\IR$ attains its maximum on
$\{1,\ldots,k\}$ at $k$.

Suppose $\R$ is ironed on interval $i < k < j$.  Then $\IR(k)$ is the
convex combination of $\R(i)$ and $\R(j)$ as
\begin{align}
\label{eq:IR(k)}
\IR(k) &= \tfrac{j-k}{j-i} \R(i) + \tfrac{k-i}{j-i}\R(j)\\
\notag
       &\leq \R(i) + \tfrac{k}{j} \R(j)
       \leq \R(i) + \R(k)
       \leq 2 \VCGr(\vals).
\end{align}
The second inequality follows because $\vali[k] \geq \vali[j]$
implies that $k \vali[k] \geq \tfrac{k}{j} j \vali[j]$; the final equality follows because both $\R(i)$ and $\R(k)$ are feasible revenues for VCG with reserve.

The above bound is almost tight.  Let $\vali[1] = k$ and
$\vali[2]=\vali[3] = \cdots = \vali[n] = 1$.  The VCG-with-reserve
benchmark obtains revenue $k$ whereas $\IR(k)= \frac{n-k}{n-1} k +
\frac{k-1}{n-1} n$, from equation~\eqref{eq:IR(k)}, which approaches
$2k-1$ as $n\to\infty$.
\end{proof}


\begin{lemma}
\label{l:ICO>>VCGreserve}
There exists a distribution $\dist$ and $n$-agent matroid environment
for which VCG with any reserve price is an $\Omega(\log n / \log \log
n)$-approximation to the Bayesian optimal mechanism for $\dist$.
\end{lemma}
\begin{proof}[Proof sketch] (A full proof is in Apprendix~\ref{app:ICO>>VCGreserve})
We construct a set system and an irregular distribution with a jagged
revenue curve that has $m=\Omega(\frac{\log n}{\log\log n})$ deep
``trenches'', such that the following are true.  (1) the Bayesian
optimal mechanism for $\dist$ gets about $1/m$ fraction of its revenue
from each trench via ironing.  (2) VCG with reserve gets similar
amount of good revenue from at most one of the trenches by setting an
appropriate reserve price, but only gets low revenue from the other
trenches due to the lack of ironing.  In total, VCG with reserve only
gets about $1/m$ fraction of the Bayesian optimal revenue.
\end{proof}

We conclude this section with a summary.  For multi-unit environments
we have given Bayesian justification for the envy-free benchmark $\EFO
\super 2$.  A mechanism that approximates this prior-free benchmark
simultaniously approximates the Bayesian optimal mechanisms for most
i.i.d.~ditributions (i.e., those satisfying the relatively
unrescrictive tail-regularity assumption).  In the next section we
will give a reduction from matroid permutation and position
environments to multi-unit environments.  This reduction implies that
the envy-free benchmark is also Bayesian justified for these
environments.  While we have not given formal justification for the
envy-free benchmark in downward-closed permutation environments
(because of Lemma~\ref{l:IC>EF}); we believe its approximation is
still of interest.


\section{Multi-unit, Position, and Matroid Permutation Environments}

\label{s:reduction}

In this section we consider matroid permutation, position auction, and
multi-unit environments.  We show that for both incentive-compatible
mechanism design and envy-free outcomes, these environments are
closely related.  In fact, for either IC or EF, respectively, the
optimal mechanisms across these environments are the same and
approximation mechanisms give the same approximation factor.  As an
example, we will focus on approximating the envy-free benchmark $\EFO
\super 2$ (Definition~\ref{d:EFO2}) with a prior-free mechanism.  Our
solution will be via a two-step reduction: we reduce matroid
permutation to position auction environments, which we then reduce to
multi-unit environments.

Recall that in a multi-unit environment it is feasible to serve any
set of agents of cardinality at most some given $k$.  In position
auction environments there are weights $w_1 \geq w_2 \geq \cdots \geq
w_n$ for positions and feasible outcomes are partial assignments of
agents to positions.  In matroid permutation environments there is a
feasibility constraint given by independent sets of a matroid, but the
roles of the agents are assigned by random permutation.

The property of these three environments that enables this reduction is
that in each environment the greedy algorithm on ironed virtual values
(with ties broken randomly) obtains the maximum ironed virtual
surplus.  The greedy algorithm works as follows: order the agents by
ironed virtual value and serve each agent in this order if her ironed
virtual value is positive and if doing so is feasible given the set of
agents previously served.  Notice that the only information needed to
perform this maximization is the ordering on the agents' ironed
virtual values (but not their magnitudes).

\begin{definition}
The {\em characteristic weights} $w_1 \geq w_2 \geq \cdots \geq w_n$
of a matroid environment are as follows: choose any valuation profile
$\vals$ with all distinct values, assign the agents to elements in the
matroid via a random permutation, run the greedy algorithm
w.r.t.\ $\vals$, and define $\wali$ to be the probability that $i$th
largest valued agent is served.
\end{definition}

\subsection{Reduction for Ironed Virtual Surplus Maximizers}

We first show ironed virtual surplus optimization in the three
environments is equivalent.

\begin{proposition}
\label{prop:ivv_reduction}
The ironed virtual surplus maximizing assignment (and its virtual surplus) 
is equal in expectation in the following environments:
\begin{enumerate}\setlength{\itemsep}{0in}
\item a matroid permutation environment with characteristic weights $\wals$, 

\item a position auction environment with weights $\wals$,

\item a convex combination of multi-unit environments where $k$ units
  are available with probability $\wali[k]-\wali[k+1]$ for
  $k\in\{1,\dots,n\}$ and $w_{n+1}=0$.
\end{enumerate}
\end{proposition}

\begin{proof}
Fix a tie-breaking rule, which induces an ordering on the agents.
Consider the greedy algorithm on the agents with non-negative $\ivv$
values according to this ordering.  The $j$-th agent with non-negative
$\ivv$ value in this ordering (1) gets allocated with probability
$w_j$ in the matroid permutation setting by definition of
characteristic weights, (2) gets assigned to position $j$ in the
position auction and hence gets allocated with probability $w_j$, and,
(3) gets allocated in $k$-unit auction for each $k\geq j$, and hence
has probability $\sum_{k\geq j}(\wali[k]-\wali[k+1])=\wali[j]$ of
being served in the convex combination setting.  Taking expectation
over all tie-breaking orders, agent $i$ has the same probability of
being served in the three settings.
\end{proof}

The following corollary is immediate.  

\begin{corollary} 
For any valuation profile $\vals$ and any weights $\wals$, the envy-free
 optimal revenue is the same in each of the environments of
Proposition~\ref{prop:ivv_reduction}.
\end{corollary}

A basic fact about incentive compatibility is that it is closed under
convex combination, i.e., a randomization over two incentive
compatible mechanisms is incentive compatible: truthtelling is an
optimal strategy in each, and so it remains an optimal strategy in the
combination.

We now illustrate how to use Proposition~\ref{prop:ivv_reduction} to show
that an incentive compatible prior-free approximation
mechanisms for multi-unit environments can be adapted to give the
same approximation factor in position auction and matroid permutation 
environments.  Consider the following incentive compatible mechanism.

\begin{definition}
\label{d:RSEM}
The {\em Random Sampling Empirical Myerson ($\RSEM$)} mechanism does
the following: (discussion of payments omitted)
\begin{enumerate}\setlength{\itemsep}{0in}
\item randomly partition the population of agents $N=\{1,\ldots,n\}$
  into two sets by flipping a fair coin for each agent,

\item designate the set containing the highest-valued agent as the
  market $M$ and the other set as the sample $S$,

\item calculate the ironed virtual surplus function $\ivv[S]$ for the
  sample $S$, and,

\item \label{step:RSEM-serve} serve a feasible subset of $M$ to
  maximize ironed virtual surplus with respect to $\ivv[S]$ and
  reject all other agents.
\end{enumerate}
\end{definition}

\begin{lemma} 
\label{l:RSEM-IC} In any single-dimensional agent environment, 
RSEM is incentive compatible.
\end{lemma}

\begin{proof}
Notice that RSEM is monotone: An agent in $S$ loses unless she
raises her bid to beat the highest-valued agent (in which case the
roles of $S$ and $M$ are reversed).  An agent in $M$ wins when the
virtual surplus maximizing set contains the agent.  If she raises her
bid, she (weakly) increases her virtual value thus increasing the
virtual surplus of any set containing her, while the virtual surplus
of other sets remain the same.  Therefore, she continues to win.  By
Theorem~\ref{t:IC_payment} monotonicity implies that, with the
appropriate payments, RSEM is incentive compatible.
\end{proof}

The proof of the following theorem is technical and we defer
discussion of it to Section~\ref{s:multi-unit-rsem}.

\begin{theorem}
\label{thm.rsem}
\label{t:multi-unit-rsem}
In multi-unit environments, RSEM is a
$\RSEMmultiunitratio$-approximation to the envy-free benchmark $\EFO
\super 2 (\vals)$.
\end{theorem}

Notice that this mechanism can easily be generalized to other
downward-closed environments. It remains incentive compatible for
these environments because it is essentially an ironed virtual surplus
optimizer on the set $M$, and furthermore, it is incentive compatible
even if the permutation that assigns agents to the set system is
fixed.  As a final corollary of Proposition~\ref{prop:ivv_reduction},
we can view RSEM's revenue in the matroid permutation or position
auction environment as the analogous convex combination of its revenue
in multi-unit auction environments.

\begin{corollary}
\label{c:RSEM}
In matroid permutation and position environments, RSEM is a prior-free
$\RSEMmultiunitratio$-approximation to the envy-free revenue $\EFO \super
2 (\vals)$.
\end{corollary}

\subsection{General Reduction}

The following prior-free approximations are essentially the best known
for digital-good and multi-unit environments.  Notably, the mechanism
from Corollary~\ref{c:multi-unit-6.5} below, is not based on ironed
virtual surplus maximization and therefore Proposition~\ref{prop:ivv_reduction}
cannot be applied to a construct matroid permutation or position
auction mechanism from it.

\begin{lemma} \citep{II-10} 
\label{l:HM-05}
In the digital good environment, there is a prior-free incentive
compatible $\digitalgoodratio$-approximation to $\EFO \super 2 (\vals)$.
\end{lemma}

We now give an approximate reduction from multi-unit auctions to
digital good auctions.  This construction and the proof that the
resulting mechanisms incentive compatibility are standard.  See, e.g.,
\citet{mye-81}, \citet{GHKSW-06}, and \citet{AH-06}.

\begin{definition}[Multi-unit Reduction] Given a $k$-agent digital good 
auction, we construct the following $k$-unit auction:
\begin{enumerate}
\item Simulate the $k$-unit Vickrey auction.
\item Simulate the $k$-agent digital good auction on the $k$ winners
  of the Vickrey auction.
\item Serve the agents who win in both stages and charge them the
  maximum of their simulation payments; reject all other agents.
\end{enumerate}
\end{definition}

\begin{theorem}
Given any digital-good auction that $\beta$-approximation to the
envy-free benchmark (resp.~Bayesian optimal auction), the multi-unit
auction from the reduction is a $2\beta$-approximation to the
envy-free benchmark (resp.~Bayesian optimal auction).
\end{theorem}

\begin{proof}
The digital good auction is a $\beta$-approximation to the envy-free
benchmark on the top $k$ agents.  The envy-free benchmark on the top
$k$ agents is equal to the VCG-with-reserve benchmark for the full set
of agents (both are equal to $\max_{i\leq k}\R(i)$).
Theorem~\ref{thm.benchmarks} states that the VCG-with-reserve
benchmark is a 2-approximation to the envy-free benchmark.  Therefore,
the multi-unit auction from the reduction is a $2\beta$-approximation
to the envy-free benchmark.
\end{proof}

\begin{corollary} 
\label{c:multi-unit-6.5}
In multi-unit environments, there is an incentive compatible
prior-free $\multiunitratio$-approximation to $\EFO \super 2 (\vals)$.
\end{corollary}

We now show how to construct, from any multi-unit auction, a position
auction and matroid permutation mechanism that has the exact same
outcome (in expectation) as a convex combination of multi-unit
auctions (as in Proposition~\ref{prop:ivv_reduction}).  The challenge
here is the distinct interfaces to the environment: in multi-unit
auctions we are given a supply constraint $k$ and we need to specify a
set of at most $k$ winners, whereas in position auctions, we are given
weights and need to output a partial assignment of agents to
positions.

\begin{definition}[Position Auction Reduction]
Given $k$-unit auction mechanisms for $k\in\{1,\dots,n\}$, we
construct the following mechanism for the position auction environment
with weights $\wals$:
\begin{enumerate}\setlength{\itemsep}{0in}

\item Introduce $n$ dummy agents and $n$ dummy positions into the
  system, indexed by $\{n+1,\dots,2n\}$.  Correspondingly, we pad
  weights $\wals$ and valuation profile $\vals$ with zeros such that
  they have dimension $2n$.

\item For each $k \in \{1,\ldots,n\}$, simulate the $k$-unit auction
  on valuation profile $\vals$, and give the unallocated leftover
  units to the dummy agents arbitrarily for free. Let the resulting
  allocation of all $2n$ agents be $\allocs \super k$.

\item Calculate the probability that each agent is served in the
  convex combination: $\alloci = \sum_{k=1}^n \alloci \super k
  (\wali[k]-\wali[k+1])$, for $i\in\{1,\dots,2n\}$.

\item Solve for a set of permutation matrices $P_t\in
  \{0,1\}^{2n\times 2n}$ and nonnegative weights $r_t$ with $\sum_t
  r_t=1$ such that $\sum_t r_t\cdot P_t\cdot \wals=\allocs$.

\item With probability $r_t$, assign agents to positions according to
  the permutation specified by $P_t$.

\item Discard dummy agents and dummy position assignments.
\end{enumerate}
\end{definition}

To justify step 4, one can verify that $\wals$ majorizes $\allocs$ in
the sense that $\sum_{i=1}^k \wali[i] \geq \sum_{i=1}^k \alloci$ for
$k\in\{1,\dots,2n\}$, with equality holding for $k=2n$.  Therefore by
a theorem of \citet{Rado-52}, the desired permutation matrices and
weights exist.  The following consequences are immediate.

\begin{lemma}
The resulting mechanism for position auction with weights $\wals$
obtained from the above reduction has the same allocation as the
convex combination of $k$-unit auctions with $(\wali[k]-\wali[k+1])$'s
as probabilities.
\end{lemma}

\begin{lemma}
Given an incentive compatible multi-unit auction, the mechanism from
the position auction reduction is also incentive compatible.
\end{lemma}

\begin{definition}[Matroid Permutation Reduction] 
Given a position auction mechanism for weights $\wals$, we construct
the following mechanism for matroid permutation environment with
characteristic weights $\wals$:
\begin{enumerate}\setlength{\itemsep}{0in}

\item We run the position auction and for $i=1,\dots,n$, let $j_i$ be
  the position assigned to agent $i$, or $j_i = \bot$ if $i$ is not
  assigned a position.

\item Reject all agents $i$ with $j_i = \bot$.

\item Run the greedy algorithm in the matroid permutation environment with
  agent $i$'s value reset to $j_i$.
\end{enumerate}
\end{definition}

The following conclusions are immediate.

\begin{lemma}
The resulting mechanism for matroid permutation environment
obtained from the above reduction has the same allocation
as the position auction.
\end{lemma}

\begin{lemma}
\label{l:mpr-ic}
Given an incentive-compatible position auction, the mechanism from
the matroid permutation reduction is incentive compatible (in matroid
permutation environments).
\end{lemma}

\begin{theorem}
\label{t:x_reduction}
The factor $\beta$ to which there is a prior-free incentive-compatible
approximation of $\EFO\super 2(\vals)$ is the same for multi-unit,
position, and matroid permutation environments.
\end{theorem}

\begin{corollary}
There is a prior-free incentive-compatible $\multiunitratio$-approximation to
$\EFO\super 2(\vals)$ in matroid permutation and position
environments.
\end{corollary}

There are two weakness in the reductions implied by
Theorem~\ref{t:x_reduction} in comparison to those implied by
Proposition~\ref{prop:ivv_reduction}.  Recall that for the latter, ironed
virtual surplus maximizations are via the greedy algorithm, and so the
reductions were algorithmically trivial.  In contrast,
Theorem~\ref{t:x_reduction} requires knowledge of the characteristic
weights to run the construction, these weights may be hard to compute.  In
addition the mechanism that results from the matroid permutation
reduction is only incentive compatible if the agents are assigned to
roles in the matroid via a random permutation as suggested in the
model.  In contrast, RSEM in matroid environments is incentive
compatible without any random permutation (Lemma~\ref{l:RSEM-IC}).

\subsection{Multi-unit Analysis of RSEM}
\label{s:multi-unit-rsem}

\newcommand{\lowprice}{\lowp}
\newcommand{\highprice}{\highp}
\newcommand{\lownumber}{\lowi}
\newcommand{\highnumber}{\highi}
\newcommand{\market}{M}
\newcommand{\sample}{S}
\newcommand{\numbers}{k} 
\newcommand{\lunder}{\underline{L}}
\newcommand{\Real}{\mathbf{R}}

\newtheorem{claim}[theorem]{Claim}


In this section we prove Theorem~\ref{t:multi-unit-rsem} which shows
that in multi-unit environments RSEM (Definition~\ref{d:RSEM}) is a
$\RSEMmultiunitratio$-approximation to the envy-free benchmark.  We
further assume that the number of units available is $k\geq 2$.  This
is without loss of generality because with $k=1$ RSEM produces the
same outcome as the Vickrey auction which is optimal with respect to
the envy-free benchmark $\EFO \super 2(\vals) = \vali[2]$.

The approximation factor is the product of three terms: $2 \times
\frac43 \times \amscr \approx \RSEMmultiunitratio$.  Roughly, these
three terms come from the following steps:
\begin{enumerate}
\item \label{step:2F>EFO} We bound the performance relative to the
  VCG-with-reserve benchmark (i.e., $\max_{i \leq k} \R(i)$).  By
  Theorem~\ref{thm.benchmarks} this benchmark can be at most a factor of two
  below the envy-free benchmark.
\item \label{step:L>4/3RSEM} We use an intermediary envy-free-like
  revenue which is (usually) at most $\frac43$ of RSEM's
  incentive-compatible revenue.
\item \label{step:AMS} Finally we employ a result due to
  \citet{AMS-09} to show that the expected relative imbalance between
  the sample and the market is $\amscr$; in particular our
  intermediary envy-free-like revenue is a $\amscr$ fraction of the
  VCG-with-reserve benchmark in expectation.
\end{enumerate}

To lay out the quantities we will be discussing, recall that RSEM
partitions the agents into a sample $S$ and market $\market$ where the
highest valued agent, 1, is conditioned to be in $\market$.  The
(ironed) virtual surplus for $\ivv[S]$ on $\market$ is then optimized
by selling the $k$ units to the agents with the $k$ highest (positive)
virtual values.  Conditioned on the partitioning, let $\highp$ and
$\lowp$ denote the supremum and infimum values that correspond to the
virtual value of the $k+1$st highest agent in $\market$.  Notice that
if an agent in $\market$ bids above $\highp$ she always wins, if she
bids below $\lowp$ she always loses, and if she bids on
$(\lowp,\highp)$ she wins with some probability.  Let $\highi$ and
$\lowi$ be the number of agents in the market that are above $\highp$
and $\lowp$, respectively, and observe that $\highi \leq k < \lowi$.
This sort of outcome was termed by \citet{HR-09} as a
$\highp$-$\lowp$-lottery.  

Notice that if we had enough units it would be envy free (for
$\market$) to sell to the top $\lowi$ agents at price $\lowp$ for
revenue $L(\lowi) = \lowp \lowi$ or to sell to the top $\highi$ agents
at price $\highp$ for revenue $L(\highi) = \highp \highi$.  Envy-free
revenue is linear and therefore with only $k$ items we can linearly
interpolate between these two revenues to obtain revenue (cf.~Section~\ref{s:benchmarks}, equation \eqref{eq:IR(k)}):
\[
L(k) = \tfrac{\lowi - k}{\lowi - \highi} L(\highi) + \tfrac{k -
  \highi}{\lowi-\highi} L(\lowi).
\]
This is the envy-free intermediary we referred to in
Step~\ref{step:L>4/3RSEM}, above.\footnote{This envy-free revenue is
  not the revenue-maximal one specified by the payment identity of
  Lemma~\ref{l:ef_char}.}

The IC revenue of this $\highp$-$\lowp$-lottery, and consequently of
RSEM, is lower than the envy-free revenue given by $L(k)$.  The top
$\highnumber$ agents each receive a unit and pay $\highprice -
(\highprice- \lowprice)\frac{(k-\highnumber+1)}{\lownumber -
  \highnumber +1}$ (by Theorem~\ref{t:IC_payment}).  The next
$\lownumber - \highnumber$ agents each receive a unit with probability
$ \frac{ k- \highnumber}{\lownumber - \highnumber}$ and pay
$\lowprice$.  The revenue of RSEM for selling $k$ units to $M$ is
thus,
\[ \RSEM(k) = (k-\highnumber)\lowprice + \highnumber \left( \highprice - (\highprice-\lowprice )\tfrac{(k-\highnumber+1)}{\lownumber - \highnumber +1} \right ) .\] 
This is also a linear function of $k$ and, while we require by
assumption that $\highi \leq k < \lowi$, formulaically $\RSEM(\lowi) =
L(\lowi)$.

We now derive a linear function that upper bounds $\IR[\sample](k)$,
the ironed revenue curve for the sample.  This upper bound will be
parameterized by the degree of imbalance between the sample and
market.

\begin{definition}
\label{d:imbalance}
The {\em imbalance} $\imbal$ of a partitioning $(S,M)$ of agents
$N=\{1,\dots,n\}$ is $\imbal = \max_{i\in N} \frac{Y_i}{i-Y_i}$ where
$Y_i=\left|\{1,2,\ldots,i\}\cap S\right|$ is the number of the highest
$i$ agents in $S$.
\end{definition}

For a given imbalance $\imbal$ the number of agents in $S$ with value
at least $\highp$ (resp.~$\lowp$) is at most $\highi \imbal$
(resp.~$\lowi\imbal$).  Define $H(\lowi) = \lowp \lowi \imbal$,
$H(\highi) = \min(\highp \highi \imbal,H(\lowi))$, and $H(k)$ as the
linear interpolation between points $(\highi,H(\highi))$ and
$(\lowi,H(\lowi))$.  The reason we take the minimum in the second step
is that since $\IR[\sample](k)$ is monotone, an upper bound on $\IR[\sample](\lowi)$ is an upper bound for all $k \in [i,j]$.  
For this definition, $H(k) \geq \IR[\sample](k)$; we defer proof of this lemma for later.

\begin{lemma}
\label{l:H>EF[S]}
For $H(k)$ defined above, $H(k)$ is at least the envy-free revenue of
the sample, $\IR[\sample](k)$.
\end{lemma}

Now we have three linear functions defined on $[\highi,\lowi]$:
$H(\cdot)$, the upper bound on the envy-free revenue from the sample;
$L(\cdot)$, the envy-free intermediary for the market; and
$\RSEM(\cdot)$, the incentive compatible revenue of RSEM on the
market.  Simply, when $L(\cdot)$ is increasing (so $H(\cdot)$ is
parallel to it), the maximal ratio between $\RSEM$ and $H$ occurs at
$k = \highi$.  The following lemma shows that this ratio is $\frac43
\imbal$ which follows from the ratio of $\RSEM(\highi)$ to $L(\highi)$
being $\frac43$; the proof is deferred to later.
\begin{lemma}
\label{l:increasing}
With imbalance $\imbal$ and increasing $L(\cdot)$, the ratio of
$\RSEM$ to $H$ is at most $\frac43 \imbal$.
\end{lemma} 
When $L(\cdot)$ is decreasing then $H(\cdot)$ is constant and the
ratio between $\RSEM$ and $H$ is maximized at $k \in
\{\highi,\lowi\}$.  Because the ratio of $H(\lowi)$ to $L(\lowi) =
\RSEM(\lowi)$ is $\imbal$ by definition; we must only show that the
$\RSEM(\highi) \geq \frac34 L(\lowi)$.  We do this in the following
lemma, proof of which is deferred to later.
\begin{lemma}
\label{l:decreasing}
With imbalance $\imbal$ and decreasing $L(\cdot)$, the ratio of $\RSEM$ to
$H$ is at most $\frac43 \imbal$.
\end{lemma}

We conclude by decomposing Step~\ref{step:AMS} into two parts. The
first is the ratio between the VCG-with-reserve benchmark on all
agents $N$ and the envy-free benchmark on the sample $S$.  The second
is the ratio between our upper bound $H(\cdot)$ on the envy-free
revenue of the sample $S$ to the revenue of $\RSEM(\cdot)$.  The
latter is, by Lemmas~\ref{l:increasing} and~\ref{l:decreasing}, at
most $\frac43 \imbal$.  To calculate the former, the VCG-with-reserve
benchmark sells to $\resi \leq k$ agents at price $\vali[\resi]$ for a
total revenue of $\R(\resi) = \resi\vali[\resi]$.  Of course, on the
sample $S$ it is envy-free to post price $\vali[\resi]$ as well;
such a price is accepted by $Y_\resi$ (from the definition of
imbalance, Definition~\ref{d:imbalance}) agents for a total revenue of
$\vali[\resi] Y_\resi$.  The optimal envy-free revenue for the sample
is no smaller.  We conclude:

\begin{lemma}\label{lem.sideA} The envy-free revenue of the sample satisfies
 $\IR^\sample (k) \geq \tfrac{Y_\resi}{\resi}\cdot \VCGr(k)$.
\end{lemma}

Combining the ratios from these two parts, $\frac{\resi}{Y_\resi}$ and
$\frac43 \imbal$ we conclude that our revenue is governed by the
expectation of random variable $X = \frac{\resi}{Y_\resi} \imbal$.  Its 
expectation is bounded by the following lemma:

\begin{lemma}\label{lem.AMS}\citep{AMS-09}
For all positive integers $\resi$, random variable $X =
\frac{\resi}{Y_\resi} \imbal$, for $Y_i$ and imbalance $\imbal$ as
defined in Definition~\ref{d:imbalance} satisfies:
$$\expect{{1/X}} \geq {1/\amscr}.$$
\end{lemma}

Lemmas~\ref{l:H>EF[S]}--\ref{lem.AMS} and Theorem~\ref{thm.benchmarks}
combine to prove Theorem~\ref{t:multi-unit-rsem}.  We conclude with
the proofs of the three deferred lemmas.

\begin{proof}[Proof of Lemma~\ref{l:H>EF[S]}]
Assume $H(k)$ is increasing.  We show that $H(k)$ is at least the
envy-free revenue of the sample, $\IR[\sample](k)$.  $\IR[S]$ is
ironed between values $\highp$ and $\lowp$; therefore, it has a line
segment that connects two points on the lines from the origin with
slope $\highp$ and $\lowp$.  These respective points are closer to the
origin than $(\highi \imbal, \highp \highi \imbal)$ and $(\lowi
\imbal, \lowp \lowi \imbal)$, by the definition of the imbalance
parameter $\imbal$.  Therefore, these latter points are strictly above
the revenue curve $\IR[S](\cdot)$ and furthermore so is the line
segment that connects them.  Since the revenue curve is monotone we
can shift these points to the left and they (and the line segment
between them) remain above the revenue curve.  The line $H(\cdot)$ is
constructed in this manner.  If, on the other hand, $H(k)$ not
increasing (then by definition it is constant) then all we need is the
end point $(\lowi,H(\lowi))$ to be above the ironed revenue curve for
$S$ which is implicit in the argument above.
\end{proof}

\begin{proof}[Proof of Lemma~\ref{l:increasing}]
We show that if $L(\cdot)$ is increasing then $RSEM(\highnumber) \geq
\tfrac34 L(\highnumber)$.  We assume that $\highnumber \geq 2$
(otherwise, a similar analysis for $\highnumber = 1$ and $k=2$ can be
employed).  To do this we give an upper bound on
\[ \tfrac{L(\highnumber) - \RSEM (\highnumber) } { L(\highnumber)} 
= \tfrac{\highprice - \lowprice}{\highprice ( \lownumber -\highnumber
  +1)} \leq \tfrac 1 6  .\] The lemma follows easily from this upper bound.  (In fact this gives a bound of $\tfrac 5 6$; the ratio $\tfrac 3 4 $ is tight for a different case.) Consider maximizing the above quantity subject to the constraints
\begin{itemize}
\item $2 \leq \highnumber < \lownumber$.  ($\highnumber$ and $\lownumber $ are integers.)
\item $\highprice > \lowprice > 0$. 
\item $\highprice \highnumber \leq \lowprice \lownumber$. 
\end{itemize} 
Rewriting the last of these as $\lowprice/\highprice \geq \highnumber /\lownumber  $, we get that 
\[ \left(1 - \tfrac{\lowprice}{\highprice} \right)\tfrac{1}{ \lownumber -\highnumber +1} \leq  
\left(1 - \tfrac{\highnumber}{\lownumber}\right) \tfrac{1}{ \lownumber -\highnumber +1} 
 = \tfrac{\lownumber - \highnumber} {\lownumber( \lownumber -\highnumber +1 )} .\] 
For any fixed $\lownumber$, the above quantity is maximized at the smallest $\highnumber$, which is 2. 
So we need to maximize 
\[ \tfrac{\lownumber - 2} {\lownumber( \lownumber -1 )} .\] 
over all integers $\lownumber \geq 3$.  It is easy to see that it is
maximized at $\lownumber = 3$ with value $1/6$.
\end{proof}

\begin{proof}[Proof of Lemma~\ref{l:decreasing}]
When $L(\cdot)$ is decreasing we show that $RSEM(\highnumber) \geq
\tfrac34 L(\lownumber)$.  Suppose that actually $L (\highnumber) >
(\lownumber +1)\lowprice.$ (This is slightly stronger than the
hypothesis in the lemma, since $L (\lownumber) =\lownumber\lowprice.$)
$\RSEM$ as a function of $k$ is linear and has a negative slope, so it
decreases in the interval $[\highnumber, \lownumber)$.  If the linear
  function were to be extended to $\lownumber$, then it would be
  exactly $L (\lownumber)$.  Thus $\RSEM(i) \geq L
  (\lownumber)$. (There is no loss of a factor.)

The other case is where $L(\lownumber) < L (\highnumber) < (\lownumber
+1)\lowprice.$ In this case one can compare $\RSEM(\highnumber)$ with
$L(\lownumber)$ and show that they are within a factor of
$(\lownumber+1)/\lownumber$ which is at most 4/3, since $\lownumber \geq 3$.
\end{proof}

\newcommand{\Xcomment}[1]{}
\Xcomment{

\begin{lemma}
For random partitioning $(S,M)$ with $1 \in M$ and $X = \imbal$ is
$Y_i=\left|\{1,2,\ldots,i\}\cap \sample\right|$
\end{lemma}

is the number of bidders in $\sample$ among the top $i$ bidders from $N$.
Let $ z = \min_i \{(i-Y_i)/Y_i\}\cup\{1\}$. 
$z$ captures the ``balance factor", that is,  it gives a lower bound 
on the ratio of the number of bidders above a certain price in $\market$ to the number of 
bidders above the same price in $\sample$.



The proof of this theorem is broken down into Lemmas \ref{lem.swap},  \ref{lem.sideA}
and Lemma \ref{lem.AMS} (together with an application of Theorem \ref{thm.benchmarks}).
Lemma \ref{lem.swap} relates the revenue from $\market$ to the ironed revenue curve 
of $\sample$.  

This relation is in terms of a random variable, $z$, which we now define. 

Recall that $Y_i=\left|\{1,2,\ldots,i\}\cap \sample\right|$

is the number of bidders in $\sample$ among the top $i$ bidders from $N$.
Let $ z = \min_i \{(i-Y_i)/Y_i\}\cup\{1\}$. 
$z$ captures the ``balance factor", that is,  it gives a lower bound 
on the ratio of the number of bidders above a certain price in $\market$ to the number of 
bidders above the same price in $\sample$. 
\begin{lemma}\label{lem.swap}
For all $k \geq 2$, $\RSEM(k) \geq z\IR^S(k) / \tfrac{4}{3}.$
\end{lemma}

Lemma \ref{lem.sideA} relates the ironed revenue curve of $\sample$ to
$\VCGr(\vals_N).$ We use $\ftwo(k)$ as a short-hand for
$\VCGr(\vals_N)$ for the rest of this section.  This relation is in
terms of $\resi$, the number of units sold by $\ftwo(k)$.  I.e.
$\ftwo(k) = \resi \valith[\resi]$.

Let $ X = z \frac{Y_\resi}{\resi}$. 
A lower bound on the expectation of $X$ for all possible values of $\resi$ is given in 
Lemma \ref{lem.AMS} (and was proved originally in \citet{AMS-09}).

\begin{proof} (of Theorem \ref{thm.rsem}.)
If $k =1$ then the theorem  is trivial. So assume $k\geq2$. 

Putting together  Lemmas \ref{lem.swap} and \ref{lem.sideA}  we get 
\[ \RSEM(k) \geq \ftwo(k) z \tfrac{Y_\resi}{\resi} / \tfrac{4}{3}  
= \ftwo(k) X / \tfrac{4}{3}.  \]
Lemma \ref{lem.AMS} gives a lower bound of $ \frac{1}{\amscr}$ on 
$ \expect[]{X}$. 
Finally, Theorem \ref{thm.benchmarks} shows that $\ftwo(k)$ is within a factor of two of  $\EFO^{(2)}$. 
\end{proof}

We still need to prove Lemma \ref{lem.swap}. We need some new definitions before we 
do that. 
 
The RSEM auction when thought of as an auction run only on  $\market$, 
 is either a $\highprice,\lowprice $ lottery or Vickrey with reserve price. 
The case when it is Vickrey with reserve price is easy, so we assume that it is a $\highprice,\lowprice $ lottery. 
Further, $\highprice$ and $\lowprice$ are consecutive bids from $\sample$. 
Let  $\highnumber$ be the number of bidders in $\market$ with valuation at least $\highprice$. 
Similarly let $\lownumber$ be the number of bidders in $\market$ with valuation at least $\lowprice$. 
$k$ is such that $\highnumber \leq k < \lownumber$. 
The top $\highnumber$  bidders get 1 unit with probability 1 and pay 
$\highprice  - \frac{(\highprice- \lowprice)(k-\highnumber+1)}{\lownumber - \highnumber +1}$. 
The next $\lownumber - \highnumber$ bidders get 1 unit with probability 
$ \frac{ k- \highnumber}{\lownumber - \highnumber}$ and pay $\lowprice$.  

The revenue, $\RSEM(k) $  is 
\[ (k-\highnumber)\lowprice + \highnumber \left( \highprice - \tfrac{(\highprice-\lowprice )(k-\highnumber+1)}{\lownumber - \highnumber +1} \right ) .\] 

Consider the line in 2 dimensional plane joining $(\highnumber, \highnumber \highprice)$ 
and $(\lownumber, \lownumber \lowprice)$. We think of this as a linear function on the interval 
$[ \highnumber, \lownumber]$, say $L$. 
$L( \highnumber ) = \highnumber \highprice$, 
$L( \lownumber ) = \lownumber \lowprice$
and $L$ is linearly interpolated in between. 

Define function $\runder(z\numbers) : = z \IR^\sample(\numbers) $ for all $\numbers$ in the interval 
in which $\IR^\sample$ is increasing. Extend it to all of $\Real_+$ by keeping it constant 
for the rest of $\Real_+$, 
that is, so that $\runder$ is monotonically non-decreasing. 
Since $\IR^S$ is concave, so is $\runder$. 

\begin{lemma}\label{claim.rundergerbar}
$\runder(\numbers) \geq z \IR^\sample(\numbers)$ for all $\numbers \in \Real_+$. 
\end{lemma}
\begin{proof}
\[ \runder(\numbers) \geq \runder(z \numbers) = z \IR^\sample(\numbers).\]
The first inequality is by monotonicity of $\runder$. The equality is by definition. 
\end{proof} 

\begin{lemma}\label{claim.lgerunder}
$ L(\numbers) \geq \runder(\numbers)$ for all $\numbers \in [\highnumber,\lownumber]$. 
\end{lemma}
\begin{proof}
By definition, $\runder(z \highnumber') = z \highnumber' \highprice $ where
$\highnumber'$ is the number of bidders in $\sample$ that are above $\highprice$. 
Similarly $\runder(z \lownumber') = z \lownumber' \lowprice $ where
$\lownumber'$ is the number of bidders in $\sample$ that are above $\lowprice$. 
From the definition of $z$, $ z \highnumber' \leq \highnumber$ and 
$ z \lownumber' \leq \lownumber$. 
Define a function $\lunder$ so that it is linear and agrees with $\runder$ in the interval 
$[z \highnumber', z\lownumber']$. 
It is enough to show that $\lunder (\lownumber) \leq L(\lownumber)$. 
\[ \lunder(\lownumber) =  \tfrac{1}{\lownumber' -\highnumber'} 
\left[ \highnumber'\lownumber' z (\highprice - \lowprice) + \lownumber ( \lownumber'\lowprice -  \highnumber'\highprice) \right] \] 
\[ \leq \tfrac{\lownumber}{\lownumber' -\highnumber'} 
\left[  \highnumber' (\highprice - \lowprice) +   \lownumber'\lowprice -  \highnumber'\highprice \right]  = \lownumber \lowprice = L(\lownumber). \] 
\end{proof}

\begin{proof} (of Lemma \ref{lem.swap}.)
By Lemma \ref{claim.rundergerbar} it is enough to show that 
$\RSEM(k) \geq 3\runder(k)/4$. 
In Lemmas \ref{lem.increasing} and \ref{lem.decreasing}, we show the following two  statements respectively. 
\begin{itemize}
\item If $L(\highnumber) \leq L(\lownumber)$ then $RSEM(\highnumber) \geq 3 L(\highnumber)/4$. 
\item If $L(\highnumber) > L(\lownumber)$ then $RSEM(\highnumber) \geq 3 L(\lownumber)/4$.
\end{itemize}
We now argue that these along with Lemma \ref{claim.lgerunder}  are sufficient. 
First of all, if $\RSEM$ is extended to $\lownumber$ by linear interpolation, its value at $\lownumber$ 
would be exactly $L(\lownumber) = \lownumber \lowprice$. 
We consider 3 cases. 

\noindent {\bf Case 1:} $\RSEM$ is decreasing in $[\highnumber, \lownumber)$. 
Note that $\runder$ is always non-decreasing. So the worst ratio of $\RSEM$ to $\runder$ 
occurs at $l-1$. But since $\RSEM(\lownumber) = L(\lownumber)$, 
 $\RSEM(\lownumber-1) \geq \runder(\lownumber-1)$. (There is no loss of a factor here.)

\noindent {\bf Case 2:} Both $\RSEM$ and $L$ are increasing in $[\highnumber, \lownumber)$. 
In this case, we actually show $\RSEM \geq 3L/4$, and we then use Lemma \ref{claim.lgerunder}. 
Since both $\RSEM$ and $L$ are linear and they coincide at $l$, the worst ratio happens  at $h$. 
This case is covered by Lemma \ref{lem.increasing} .

\noindent {\bf Case 3:} $\RSEM$ is increasing and $L$ is decreasing in $[\highnumber, \lownumber)$. 
In this case we compare the lowest point of $\RSEM$ which is $\RSEM(h)$ with the highest 
point of $\runder$, which is $\runder(l)$. This is covered by Lemma \ref{lem.decreasing}. 
\end{proof}

}

\section{Downward-closed Permutation Environments}

\label{s:random_sampling}

In this section, we will show that a variant of RSEM (recall
Definition~\ref{d:RSEM}) approximates the envy-free benchmark by a
constant factor.

\begin{definition}[RSEM$'$]
\label{d:RSEM'}
The variant RSEM$'$ is identical to RSEM except Step~\ref{step:RSEM-serve}:
\begin{enumerate}\setlength{\itemsep}{0in}
\item[\ref{step:RSEM-serve}$'$.] find the feasible subset $W$ of $N$
  (the full set of agents) to maximize ironed virtual surplus with
  respect to $\ivv[S]$, serve agents in $M \cap W$ (the winners
  from the market $M$) only, and reject all other agents.
\end{enumerate}
\end{definition}

The proof we give that RSEM$'$ is a good approximation to the
envy-free benchmark is based on the fact that with large probability
the sample and market satisfy a natural balanced condition.  This
condition requires that, for all prefixes of the agents sorted by value,
a good fraction of these agents are in each of the market and
sample.  The proof then has three main steps: show the probability of
balance is high, show balance implies that the IC revenue of RSEM$'$ (on
the market $M$) is close to the optimal EF revenue for the sample $S$,
and show the expected optimal EF revenue of the sample is close to the
envy-free benchmark.

\begin{theorem}
\label{t:IC-RESM-dc}
For downward-closed permutation environments, $\expect{\IC{\RSEM'}(\vals)} \geq
\frac{1}{\RSEMdownwardclosedratio} \EFO \super 2 (\vals)$  for all $\vals$.
\end{theorem}

\def\balance{{\cal B}}
\begin{proof}
Let $\balance$ denote the event that the market and sample are
balanced.  Lemma~\ref{l:balanced2} states that the probability that
the partition is balanced is at least:
\begin{align*}
\prob{\balance} &\geq 0.339.\\
\intertext{The expected IC revenue of RSEM$'$ is at least its revenue
  conditioned on the partitioning into sample and market being
  balanced (denoted as event $\balance$).  I.e.,}
\expect{\IC{\RSEM'}(\vals)} & \geq  \prob{\balance}\expect{\IC{\RSEM'}(\vals) \given
  \balance}.\\ 
\intertext{Lemma~\ref{l:balance-32} states that for any $\vals$ that
  balance implies the IC revenue of RSEM$'$ is at least a $\frac{1}{32}$
  fraction of the EF optimal revenue on the sample.  Taking
  expectations,}
\expect{\IC{\RSEM'}(\vals) \given \balance} & \geq \tfrac{1}{32} \expect{\EFO(\vals_S) \given \balance}.\\ 
\intertext{Lemma~\ref{l:EFO_S-2} states that the EF optimal revenue
  on the sample is at least half the envy-free benchmark, in
  expectation and conditioned on a balanced partitioning.  I.e.,}
\expect{\EFO(\vals_S) \given \balance} & \geq \tfrac{1}{2} \EFO\super 2(\vals).\\
\intertext{Combining the above inequalities we conclude that the IC revenue of RSEM$'$ is a least a $\RSEMdownwardclosedratio$-approximation to the envy-free benchmark.  I.e.,}
\expect{\IC{\RSEM'}(\vals)} &\geq \tfrac{1}{\RSEMdownwardclosedratio} \EFO \super 2 (\vals). \qedhere
\end{align*}
\end{proof}

\subsection{Balanced Partitioning}

We now show that with high probability the partitioning of the agents
into the market and sample satisfies a natural balanced property.
Recall that, by definition, agent 1 is in $M$.  This balanced property
is a double-sided version of the balanced property introduced by
\citet{FFHK-05}.

\begin{definition}
A partitioning $(S,M)$ of agents $N=\{1,\dots,n\}$ is {\em balanced}
if $1 \in M$ and $2 \in S$ and for any set of three or more of the
highest valued agents both the market and sample contain at least a
quarter of agents in the set.  I.e., for $i \geq 3$, $\setsize{S \cap
  \{1,\ldots,i\}} \geq i/4$ and $\setsize{M \cap \{1,\ldots,i\}} \geq
i/4$.
\end{definition}


\begin{lemma}
\label{l:balanced2}
Conditioning on $1\in M$, a random partitioning $(S,M)$ of $N$ is
balanced with probability at least $0.339$.
\end{lemma}

\begin{proof}
Conditioning on $1 \in M$ and $2 \in S$, the probability that either
part is imbalanced can be calculated to be at most $0.161$ by a simple
probability of ruin analysis which comes from \citet{FFHK-05} (details
given below).  By the union bound, both parts are balanced with
probability at least $0.678$.  Agent 2 is in $S$ with probability
$1/2$ so the probability of balance conditioned on agent 1 in $M$ is
at least $0.339$.

The following analysis from \citet{FFHK-05} shows that the probability
that $S$ is imbalanced is at most $0.161$.  Consider the random
variable $Z_i = 4\setsize{S \cap \{1,\ldots,i\}} - i$; the balanced
condition is equivalent to $Z_i \geq 0$ for all $i \geq 3$.  By the
conditioning $i=2$ and $S \cap \{1,2\} = \{2\}$ imply that $Z_2 = 2$.
View $Z_i$ as the positions of a random walk on the integers that
starts from position two and takes three steps forward (at step $i$ with $i
\in S$) or one step back (at step $i$ with $i\not\in S$), each with
probability one half.  If this random walk ever arrives at position
$-1$ the partitioning is imbalanced.  This probability $r$ of ever
taking one step back in such a random walk can be calculated as the
root of $r^4-2r+1$ on interval $(0,1)$ which is about $0.544$.  The
probability of imbalance is then $r^3 \leq 0.161$ (i.e., if we ever take
three steps back when starting from position two).  By symmetry, the
probability of imbalance in the market $M$ is also at most $0.161$.
\end{proof}

\subsection{Market revenue versus sample revenue}

We now show that conditioned on a balanced partitioning of the agents
into a market and sample, that the revenue of RSEM$'$ from the market
is close to the envy-free optimal revenue from the sample.  The
revenue of RSEM$'$ is precisely $\IC{S}_M(\vals_N)$, i.e., the revenue
we get from the agents in $M$ when using the virtual value functions
from $S$ and optimizing virtual values on the over the full set of
agents $N$.  We wish to compare this revenue to the envy-free optimal
revenue on the sample, $\EFO(\vals_S)$.  Our bound follows from three
steps: the IC revenue of RSEM$'$ from $M$ is close to its EF revenue
from $M$, the EF revenue of RSEM$'$ from $M$ is close to its EF
revenue from the full set of agents $N$, and this EF revenue from the
full set of agents is close to the optimal envy-free revenue from the
sample.

\begin{lemma} 
\label{l:balance-32}
Given a balanced partitioning $(M,S)$, $\IC{\RSEM'}(\vals) \geq \tfrac{1}{32} \EFO(\vals_S)$.
\end{lemma}

\begin{proof} The proof is given by the following sequence of inequalities:
\begin{align*}
\IC{\RSEM'}(\vals)
   &= \IC{S}_M(\vals_N)                   &\text{(Definition~\ref{d:RSEM'})}\\
   &\geq \tfrac{1}{2} \EF{S}_M(\vals_N)    &	\text{(Lemma \ref{l:IC>EF/2})}\\
   &\geq \tfrac{1}{8} \EF{S}_N(\vals_N)    &	\text{(below, Lemma \ref{l:dominance})}\\
   &\geq \tfrac{1}{32} \EF{S}_S(\vals_S)   & 	\text{(below, Lemma \ref{l:key})}\\
   &= \tfrac{1}{32} \EFO(\vals_S). &&\qedhere
\end{align*}
\end{proof}

The first step (Lemma~\ref{l:IC>EF/2}) of Lemma~\ref{l:balance-32} was
proven in Section~\ref{s:ic_vs_ef}.  To see the second step notice
that the envy-free payment of the agents $i \in N$ in
$\EF{S}(\vals_N)$ form a non-increasing sequence.  This non-increasing
sequence of payments can be plugged into the following lemma
(Lemma~\ref{l:dominance}) as the $a_i$s.  The proof of the lemma
follows directly from its statement.

\begin{lemma}
\label{l:dominance}
Given a balanced partitioning $(S,M)$, for every non-increasing
sequence $a_1,\ldots,a_n$ of nonnegative reals and all $i \in N$,
$\sum_{j\in M\cap\{1,\ldots, i\}}a_j\geq \frac{1}{4}\sum_{j\in
  \{1,\ldots, i\}}a_j$.
\end{lemma}

To show the third step, we need to give a detailed analysis of what
happens in terms of envy-free revenue when we optimize for the wrong
virtual values.  To do that we will define and consider the {\em
  effective revenue curve}, $\ER$, and {\em perceived revenue curve},
$\Rhat$.  Intuitively, $\Rhat$ corresponds to the revenue we think we
get when optimizing $\ivv[S]$ on $\vals$, and $\ER$ corresponds to the
revenue curve we actually end up with.



\begin{definition}[Effective revenue curve $\ER$]
For values $\vals$ and ironed virtual valuations $\ivv^{S}$ for $S$:
group agents with equal nonnegative $\ivv^{S}$ values into consecutive
classes $\{1,\ldots,n_1\}$, $\{n_1+1,\ldots,n_2\}$,
\ldots, $\{n_{t-1}+1,\ldots,n_t\}$ and define the {\em effective revenue
  curve} $\ER$ from $\R=\R[\vals]$ by connecting the points
$(0,0)$, $(n_1,\R(n_1))$, \ldots, $(n_t,\R(n_t))$ and then extending
horizontally to $(n,\R(n_t))$, i.e., ironing the values in each class.
\end{definition}

\begin{figure}[t]
\begin{center}
\subfigure[Effective Ironing]{%
        \label{fig:effective_ironing}
	\includegraphics[scale=0.45]{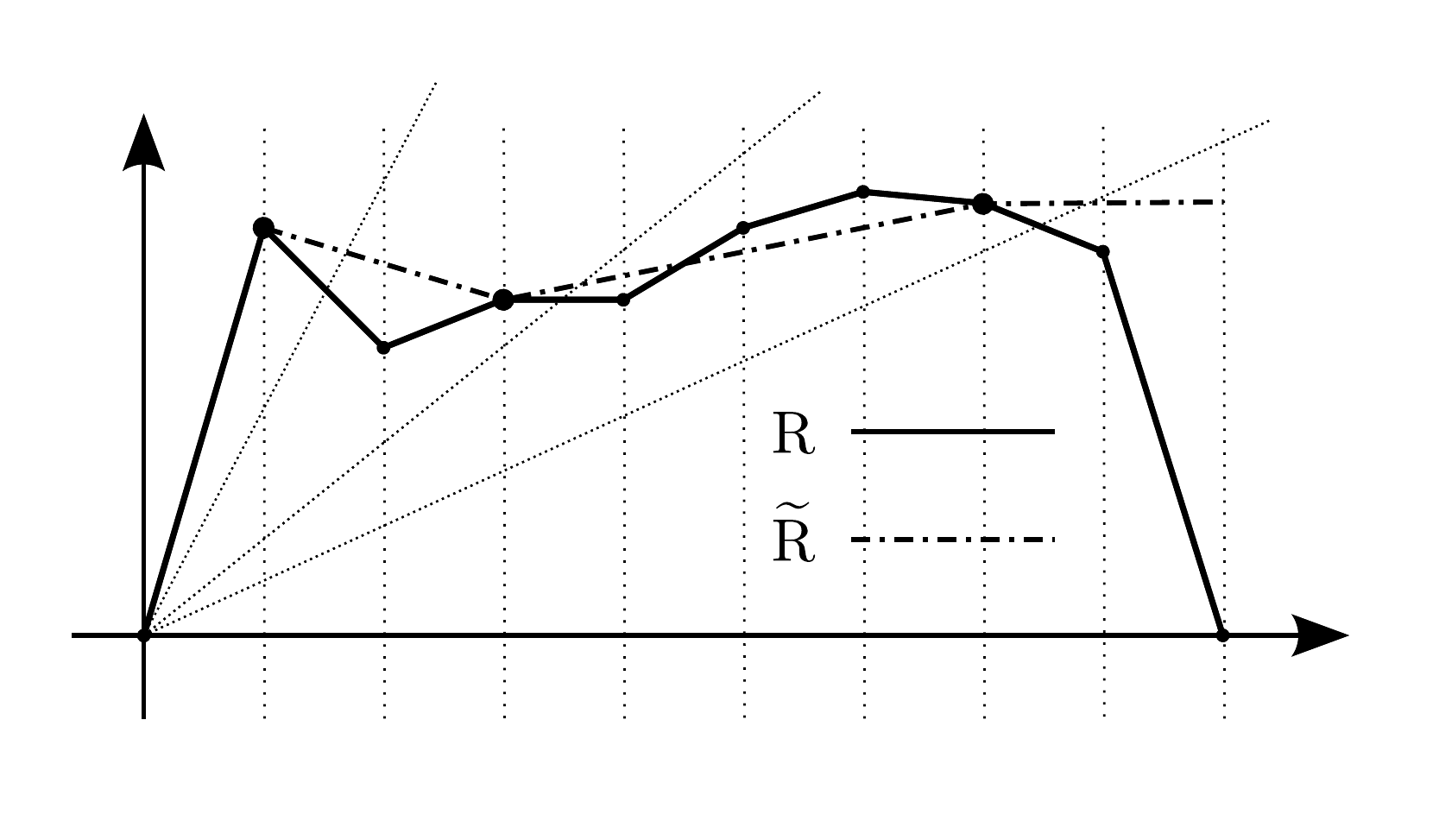}
        }
\subfigure[Revenue dominance]{%
\label{fig:curves}
	\includegraphics[scale=0.45]{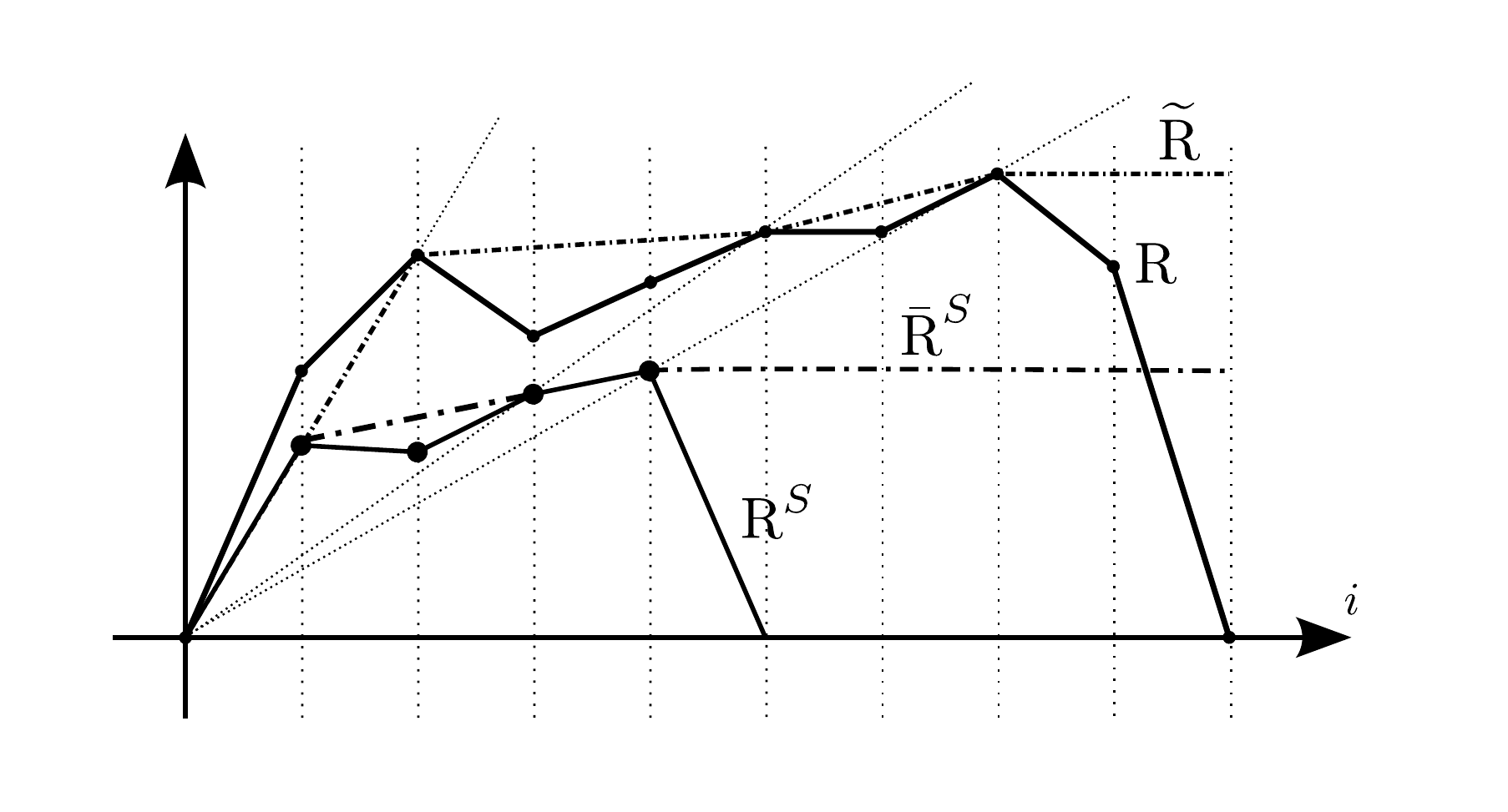}
        }
\end{center}
\caption{Effective revenue curves and revenue dominance depicted.}
\end{figure}

Figure~\ref{fig:effective_ironing} depicts an example of the effective
revenue curve.  The three rays from the origin, which correspond to
values at which $\ivv[S]$ makes a piece-wise jump, divide the first
orthant into four regions.  For every region, every point $(i, \R(i))$
in the region (which corresponds to value $\vali$) has the same
$\ivv^{S}$ value.  In each region these points get ``ironed'', 
and hence the line segment in $\ER$. 
\begin{lemma}
\label{l:effective}
$\EF{\ivv[S]}(\vals)=\sum\nolimits_{i=1}^n \ER(i)
\cdot (\alloci^{S}(\vals)
-\alloci[i+1]^{S}(\vals))$.
\end{lemma}

\begin{proof}
\begin{align*}
\EF{\ivv[S]}(\vals)
&=\sum\nolimits_{i=1}^n \R(i)
\cdot (\alloci^{S}(\vals)
-\alloci[i+1]^{S}(\vals))\\
&=\sum\nolimits_{i=1}^n \ER(i)
\cdot (\alloci^{S}(\vals)
-\alloci[i+1]^{S}(\vals))
\end{align*}
Here the first equality is by Lemma~\ref{l:ironing}.
To justify the second equality, 
note that whenever $\ER(i)\neq \R(i)$, there are two cases:
(1) $i$ is in $\{n_{j-1}+1,\ldots,n_j-1\}$ for some $j$,
and so $\vali$ and $\vali[i+1]$ have the same $\ivv^S$ value,
and hence $\alloci^{S}(\vals)=\alloci[i+1]^{S}(\vals)$; and
(2) $i$ is bigger than $n_t$, and so $\vali$ and $\vali[i+1]$ both have negative $\ivv^S$ value, and hence 
$\alloci^{S}(\vals)=\alloci[i+1]^{S}(\vals)=0$.
\end{proof}

For a set of agents $S$, let $\vals_S$ denote $(\vals_S,\mathbf{0}_{N
  - S})$, i.e., the valuation profile (of $n$ agents) obtained from
$\vals$ by decreasing the values of agents outside $S$ to $0$. Note
that $\vals=\vals_N$.  Let $\R^S$ and $\IRs$ be the revenue curve and
ironed revenue curve of the valuation profile $\vals_S$ respectively.
\begin{lemma}
\label{l:r_dominance}
For all $1\leq i\leq n$,
$\ER(i)\geq \IRs(i)$.
\end{lemma}

\begin{proof}[Proof sketch] Figure~\ref{fig:curves} depicts the relationship between the revenue curves.
Observe that revenue curve $\R$ dominates $\R^S$ in the sense that for
every slope $t$, the intersection of the ray $y=tx$ with $\R$ is
farther away from the origin than its intersection with
$\R^S$. Transforming $\R$ and $\R^S$ to the effective revenue curves
using the same ironed virtual valuation function $\ivv[S]$ do not
change such dominance relationship, and moreover, because $\IRs$ is
non-decreasing and concave, it follows that vertical dominance also
holds, i.e., $\ER(i)\geq \IRs(i)$ for all $i$.
\end{proof}

\begin{definition}[Perceived revenue curve $\Rhat$]
The {\em perceived revenue curve} for $\ivv[S]$ on $\vals$ is given by $\Rhat(i)=\sum_{j=1}^i \ivv[S](\vali)$ for $i \in N$.
\end{definition}

Let $\vhats$ be the valuation profile corresponding to $\Rhat$, i.e.,
$\vhati=\Rhat(i)/i$, and let $\allocs^{\vhats}$ be the ironed virtual
surplus maximizer for $\ivv^{\vhats}$.
\begin{lemma}
\label{l:v_hat}
$\alloci^S(\vals)=\alloci^{\vhats}(\vhats)$.
\end{lemma}
\begin{proof}
Compare running the ironed virtual surplus maximizer $\allocs^S$ for $\ivv^S$ on $\vals$ with running  $\allocs^{\vhats}$ for $\ivv^{\vhats}$ on $\vhats$, 
the ironed virtual valuation of agent $i$ in either case is equal to $\ivv^S(\vali)$. 
Therefore these two ironed virtual surplus optimizers will choose the same allocation, 
and the lemma follows.
\end{proof}


\begin{lemma}
\label{l:vhat_vs_v_S}
Given a balanced partitioning $(S,M)$, then $\IRs(i) \geq
\frac{1}{4}\Rhat(i)\geq \frac{1}{4}\IRs(i) $ for all $1\leq i\leq n$.
\end{lemma}
\begin{proof}
For each $i$, $\Rhat(i)$ is the sum of $i$ largest ironed virtual
values in $N$ with respect to $\ivv[S]$ while $\IRs(i)$ is the sum of
the $i$ largest with respect to $S$.  Therefore $\Rhat\geq\IRs(i)$.
Since $(S,M)$ is double-side balanced, applying
Lemma~\ref{l:dominance}, we also have that for all $i$, $\IRs(i)\geq
\frac{1}{4}\Rhat(i)$.
\end{proof}

Now we are ready prove the following key lemma:

\begin{lemma}
\label{l:key}
For any downward-closed permutation environments, any valuation
profile $\vals$, and balanced partitioning $(S,M)$,
$\EF{\ivv[S]}(\vals_N) \geq \frac{1}{4}
\EF{\ivv[S]}(\vals_S)=\frac{1}{4}\EFO(\vals_S)$.
\end{lemma}

\begin{proof}
Let $\allocs^S$ and $\allocs^{\vhats}$ be short-hands for
 the ironed virtual surplus optimizers with ironed virtual valuation functions 
 defined for $\vals_S$ and $\vhats$, respectively.
The proof is by the following inequalities:
\begin{align*}
\EF{\ivv[S]}(\vals_N)
&=\sum\nolimits_i \ER(i)
  \cdot(\alloci^{S}(\vals_N)
  -\alloci[i+1]^{S}(\vals_N))\\
&=\sum\nolimits_i \ER(i)
  \cdot(\alloci^{\vhats}(\vhats)
  -\alloci[i+1]^{\vhats}(\vhats))\\
&\geq \tfrac{1}{4}\cdot \sum\nolimits_i \Rhat(i)
  \cdot(\alloci^{\vhats}(\vhats)
  -\alloci[i+1]^{\vhats}(\vhats))\\
&\geq \tfrac{1}{4}\cdot \sum\nolimits_i \Rhat(i)
  \cdot(\alloci^{S}(\vals_S)
  -\alloci[i+1]^{S}(\vals_S))\\
&\geq \tfrac{1}{4}\cdot \sum\nolimits_i \IRs(i)
  \cdot(\alloci^{S}(\vals_S)
  -\alloci[i+1]^{S}(\vals_S)).
\end{align*}
Here the first two equalities are guaranteed by our definitions of
$\ER$ and $\Rhat$.  The first inequality is by
Lemma~\ref{l:r_dominance} and Lemma~\ref{l:vhat_vs_v_S}, the second
inequality is by the optimality of $\allocs^{\vhats}$ for $\vhats$,
and the third inequality is by Lemma~\ref{l:vhat_vs_v_S} again.
\end{proof}

\subsection{Expected sample revenue versus the envy-free benchmark}

We now show that the expected envy-free revenue of the sample compares
favorably with the envy-free benchmark; this is the last ingredient in
the proof of Theorem~\ref{t:IC-RESM-dc}.  We will make this argument
conditioned on a balanced partitioning; however, the result is true
for any symmetric conditioning (including none at all).

\begin{lemma}
\label{l:subadd}
For a partitioning $(S,M)$ of $N$, we have that
$\EFO(\vals_S) + \EFO(\vals_M) \geq \EFO(\vals_{N})$.
\end{lemma}
\begin{proof}
$\EFO(\vals_N)=\EFO(\vals_{S\cup M})$ is the maximum revenue we can
  get from $S\cup M$ subject to the envy free constraints.  Let agents
  in $M$ contribute total revenue $R$ to $\EFO(\vals_{N})$.  By
  setting the agents in $S$ to have zero valuations to obtain
  valuation profile $\vals_S$, we basically removed envy-freedom
  constraints between agents in $S$ and agents in $M$.  With fewer
  envy-freedom constraints, the maximum envy-free revenue we can get
  from $M$, i.e., $\EFO(\vals_M)$, can only be larger.  Similarly, the
  total revenue that $S$ contributes to $\EFO(\vals_{N})$ is at most
  $\EFO(\vals_S)$, and our lemma follows.
\end{proof}

\begin{lemma} 
\label{l:EFO_S-2}
$\expect{\EFO(\vals_S) \given \balance} \geq \tfrac{1}{2} \EFO \super
  2(\vals)$ where $\balance$ is the event that $S$ and $M$ are
 balanced.
\end{lemma}

\begin{proof} The lemma follows from Lemma~\ref{l:subadd}, symmetry, and the fact that agent 1 is always in $M$.
\end{proof}

\bibliographystyle{apalike}
\bibliography{auctions}

\appendix

\section{Proofs}

%
%
%
%
%

\subsection{Proof of Lemma~\ref{l:ICO>>VCGreserve}}

\label{app:ICO>>VCGreserve}

Fix some number $m$. 

The matroid we use is a partition matroid. In general in a partition
matroid, the ground set is partitioned into a number of disjoint sets,
or sectors, where each sector is associated with a capacity number.  A
set is feasible if and only if its intersection with each sector does
not exceed the capacity number of the sector.

Now we define the partition matroid we use.
For each
$k\in\{1,\ldots,m\}$, a type $k$ sector contains
$m^{3k-1}$ elements or agents and has capacity one.  There are $m^{2m-2k}$ disjoint type $k$ sectors for each $k\in\{1,\ldots,m\}$. So total number of
agents $n$ is at most $m^{O(m)}$. Hence $m$ is at least of order
$\frac{\log n}{\log\log n}$.

Next we define the ``sydney opera house distribution''. The
distribution $F$ is such that the value is distributed according to
uniform distribution $[m^{2k+1}-\epsilon,m^{2k+1}+\epsilon]$ with
probability $\frac{1}{m^{3k}}-\frac{1}{m^{3k+3}}$ for
$k\in\{0,\ldots,m-1\}$, and with probability $\frac{1}{m^{3k}}$ for
$k=m$. Here we take $\epsilon$ to be some sufficiently negligible
positive amount, and we will often omit $\epsilon$ related terms. So
for each $k$ the revenue function $R$ at $\frac{1}{m^{3k+3}}$ has left
limit
$R(\frac{1}{m^{3k+3}}-)=\frac{m^{2k+3}}{m^{3k+3}}=\frac{1}{m^{k}}$,
and right limit
$R(\frac{1}{m^{3k+3}}+)=\frac{m^{2k+1}}{m^{3k+3}}=\frac{1}{m^{k+2}}$.
Hence the ironed virtual valuation between quantile
$\frac{1}{m^{3k+3}}$ to quantile $\frac{1}{m^{3k}}$ is
$\frac{\frac{1}{m^{k-1}}-\frac{1}{m^{k}}}{\frac{1}{m^{3k}}-\frac{1}{m^{3k+3}}}\approx
m^{2k+1}$.  Note that the ironed virtual valuation is equal to
valuation, ignoring minor terms.

To calculate the revenue of Myerson's auction, for a type $k$ sector,
there are $m^{3k-1}$ agents. With probability at least
$1-(1-\frac{1}{m^{3k}})^{m^{3k-1}}\approx\frac{1}{m}$, the highest
agent is in quantile range $(0,\frac{1}{m^{3k}})$, with ironed virtual
valuation at least $m^{2k+1}$. So the expected ironed virtual
valuation from a type $k$ sector is at least $m^{2k}$. Multiplied by
the number of type $k$ sectors the total ironed virtual valuation, and
hence expected revenue is at least $\sum_{k}m^{2k}\cdot
m^{2m-2k}=m\cdot m^{2m}$.

To calculate the revenue of $VCG$ with some reserve $r$, suppose
w.l.o.g.\ $r\approx m^{2k'+1}$ for some $k'$. For a type $k$ sector
with $m^{3k-1}$ agents, the dominant amount of revenue is obtained
from the following two cases:
\begin{enumerate}
\item When there are at least two agents with value at least
  $m^{2k+1}-\epsilon$ (i.e. in quantile $\frac{1}{m^{3k}}$), the lower
  of which has value at most $m^{2k+1}+\epsilon$.  This happens with
  probability roughly $\frac{1}{m^{2}}$, and gives revenue
  $m^{2k+1}$. Therefore the expected revenue we get from this case is
  $\frac{1}{m^{2}}\cdot m^{2k+1}$, which multiplied by the number of
  type $k$ sector, is $O(m^{2m-1})$.
\item When $k=k'$, and there is at least one agent who beats the
  reserve $m^{2k+1}$. This happens with probability at most
  $\frac{1}{m}$.  Therefore the expected revenue from this case is
  $m^{2k+1}\cdot\frac{1}{m}$, which multiplied by the number of type
  $k$ sectors is $O(m^{2m})$.
\end{enumerate}
Summing over all $k$, the total expected revenue of $VCG$ with reserve
$r$ is at most $O(m^{2m})+m\cdot O(m^{2m-1})=O(m^{2m})$, which is less
than that of Myerson's auction by a factor of $\Omega(m)=\Omega(\log
n/\log\log n)$.
%
%


\subsection{Proof of Lemma~\ref{l:IC>EF}}

\label{app:IC>EF}

Let there be $n+1$ agents. The ``$1$ vs $n$'' set system has two maximum feasible sets, one is a singleton set and the other one has size $n$. These two sets are disjoint.
We define the valuation profile by specifying the virtual valuations.
There are $n$ ``small'' agents with virtual values $\val+\epsilon, \val+2\epsilon,\dots,\val+n\epsilon$ respectively, and one ``big'' agent with virtual value $n\val+\frac{n(n+1)}{2}\epsilon-\epsilon^2$ for some small positive $\epsilon$. The choice of the $\epsilon$ terms is such that for the sum of the virtual valuations of the first $n$ agents to beat the big agent, no small agent can lower her virtual value to some other agent's virtual value. We will ignore $\epsilon$ terms from now on.
Correspondingly, one can calculate the revenue curve, and then derive the valuations of the agents: the valuation of the big agent is $n\val$, and the small agents have values $\frac{n+1}{2}\val$, $\frac{n+2}{3}\val$,\ldots, $\frac{2n}{n+1}\val$, ignoring $\epsilon$ terms.
The allocation rule is the ironed virtual surplus optimizer w.r.t.\ this valuation profile. Note that a reserve of $\frac{2n}{n+1}\val$ is set because any value lower than this corresponds to a negative ironed virtual value.

Observe that every agent wins if and only if she is assigned to the size $n$ set, which happens with probability $n/(n+1)$. Therefore the EF revenue is  $\frac{2n}{n+1}\val \cdot \frac{n}{n+1} \cdot  (n+1)=\frac{2n^2}{n+1}\val$.
To calculate the IC revenue, with probability $n/(n+1)$, the big agent is assigned to the size $n$ set, and every of the $n$ winning agents pays the reserve $\frac{2n}{n+1}\val$. Also with probability $1/(n+1)$, the big agent is assigned to the singleton set, and every agent has to pay her own value,
which sums up to $\Theta(n\val \log(n))$. Therefore the IC revenue is $\frac{2n}{n+1}\val\cdot \frac{n}{n+1}\cdot n  +\frac{1}{n+1} \Theta(n\val \log(n))$, which is larger than EF revenue for sufficiently large $n$.

\subsection{Proof of Lemma~\ref{t:tail-regular}}
\label{app:tail-regular}
\begin{proof}[Proof]

By the reduction from matroid permutation environments to multi-unit
environments, it is sufficient to prove the statement for $k$-unit auctions.
Let $\ivv$ be the ironed virtual surplus maximizer for $F$.
We first upper-bound IC revenue from bidders 2 to $n$:
\begin{align*}
\IC{F}_{2\dots n}(\vals) & \leq \EF{\ivv}_{2\dots n}(\vals) \\
&  \leq \EFO(\vali[2],\vali[3],\dots,\vali[n], 0)\\
&  \leq \EFO(\vals\super{2}) \\
&    = \EFO\super{2}(\vals) 
\end{align*}

Here the first inequality is by Lemma~\ref{l:EF>IC}, and the last
equality is by definition of $\EFO\super{2}$.

To see the second inequality, both left hand side and right hand side
correspond to the maximum envy-free revenue from $2\dots n$ that
correspond to some outcome with at most $k$ items allocated and no
envy among $2\dots n$, except that in the right hand side, the outcome
that maximizes this revenue is chosen.

To see the third inequality, note that the revenue curve of
$\vals\super{2}$ dominates that of $(\vali[2],\vali[3],\dots,\vali[n],
0)$, and hence by Lemma~\ref{l:vv=rev}, for every allocation, the
envy-free revenue for $\vals\super{2}$ can only be higher.

Next we upper-bound IC revenue from bidder 1, where the expectation is over i.i.d.\ draws from distribution $F$.
\begin{align*}
 E[\IC{\ivv}_1(\vals)] & \leq E[\IC{\ivv}(\vals)\text{ for single item auction}]\\
              &        \leq 2 E[\vali[2]]           \\
              &        \leq 2 E[\EFO\super{2}(\vals)]   
\end{align*}
Here the second inequality is by the tail regularity assumption.  The
third inequality is because $\EFO\super{2}(\vals)\geq \vali[2]$.

To see the first inequality, consider the mechanism that first runs
$\ivv$, and then only allows the highest bidder (bidder 1) to win.
The IC payment of bidder 1 is the maximum of the second highest bid
and the threshold bid for bidder 1 to win in $\ivv$.  This is as much
as the threshold (or revenue) from bidder 1 in $\ivv$ as in the left
hand side. On the other hand, the IC revenue of this mechanism is at
most that of the optimal single-item auction, which equals to the
right hand side.

Together, we have that $E[\IC{\ivv}(\vals)] \leq 3 E[\EFO\super{2}(\vals)]$.

\end{proof}

\end{document}